\documentclass[a4paper,UKenglish,cleveref, autoref, thm-restate]{lipics-v2021}
\nolinenumbers
\hideLIPIcs

\usepackage[dvipsnames]{xcolor}
\usepackage{tikz}
\usetikzlibrary{calc,shapes}
\usepackage[utf8]{inputenc}
\usepackage{amsmath,amsthm,amssymb}
\usepackage{mathtools}
\usepackage{diagbox}
\usepackage{todonotes}
\usepackage[unicode]{hyperref}
\usepackage{xspace}
\usepackage{thm-restate}
\usepackage[noend]{algpseudocode}
\usepackage{caption}
\usepackage{subcaption}
\usepackage[algoruled,boxed,lined]{algorithm2e}
\usepackage{wrapfig}
\usepackage[normalem]{ulem}
\usepackage{multirow}
\usepackage{tabularx}

\usepackage{makecell}

\usepackage{thm-restate}
\usepackage{caption}

\usepackage{booktabs}
\usepackage{makecell}

\SetKwInOut{Parameter}{Parameters}

\newcommand{\RR}{\mathbb{R}}

\newcommand{\PP}{\mathcal{P}}
\newcommand{\TT}{\mathbb{T}}
\newcommand{\HH}{\mathcal{H}}
\newcommand{\SSS}{\mathcal{S}}

\newcommand{\pmax}{\ensuremath{p_{\max}}}
\newcommand{\Tw}{\ensuremath{w}}
\newcommand{\dd}{\ensuremath{d}}
\newcommand{\co}{\ensuremath{c}}

\newcommand{\br}[1]{\ensuremath{\{#1\}}}

\DeclareMathOperator {\bigO}{\mathcal{O}}

\newcommand{\NP}{\texorpdfstring{$\mathcal{NP}$}{NP}}

\SetKwFunction{leafNode}{computeLeafNode}
\SetKwFunction{intrNode}{computeIntroduceNode}
\SetKwFunction{forgNode}{computeForgetNode}
\SetKwFunction{intrEdgeNode}{computeIntroduceEdgeNode}
\SetKwFunction{joinNode}{computeJoinNode}

\newboolean{ShowOldText}
\setboolean{ShowOldText}{false}

\newcommand{\debug}[1]{
\ifthenelse{\boolean{ShowOldText}}{#1}{}
}

\newcommand{\jk}[1]{#1}
\newenvironment{jkblock}{}{}

\title{Parameterized Algorithms for Computing Pareto Sets}

\authorrunning{J.\,M. Könen, H. Röglin, and T. Stuck}
\Copyright{Joshua M. Könen, Heiko Röglin, and Tarek Stuck}

\ccsdesc[500]{Theory of computation~Design and analysis of algorithms}
\ccsdesc[500]{Theory of computation~Parameterized complexity and exact algorithms}
\ccsdesc[500]{Theory of computation~Graph algorithms analysis}

\keywords{parameterized algorithms, treewidth, multicriteria optimization problems, multicriteria MST, multicriteria TSP, polygon aggregation}

\supplement{We provide the source code of our algorithm as well as an interactive visualization tool for exploring all Pareto-optimal solutions.}

\supplementdetails[linktext={Bicriteria Aggregation Code}]{Software}{https://github.com/Tarek-pub/Bicriteria_Aggregation}

\supplementdetails[linktext={Bicriteria Aggregation Visualization Tool}]{InteractiveResource}{https://github.com/Tarek-pub/Bicriteria_Aggregation_plotting}

\EventEditors{Anne Benoit, Haim Kaplan, Sebastian Wild, and Grzegorz Herman}
\EventNoEds{4}
\EventLongTitle{33rd Annual European Symposium on Algorithms (ESA 2025)}
\EventShortTitle{ESA 2025}
\EventAcronym{ESA}
\EventYear{2025}
\EventDate{September 15--17, 2025}
\EventLocation{Warsaw, Poland}
\EventLogo{}
\SeriesVolume{351}
\ArticleNo{104}

\author{Joshua Marc K\"onen}{Institute of Computer Science, University of Bonn, Germany}{koenen@cs.uni-bonn.de}{https://orcid.org/0000-0003-4245-4812}{}
\author{Heiko R\"oglin}{Institute of Computer Science, University of Bonn, Germany}{roeglin@cs.uni-bonn.de}{https://orcid.org/0009-0006-8438-3986}{}
\author{Tarek Stuck}{Institute of Computer Science, University of Bonn, Germany}{s6tastuc@uni-bonn.de}{}{}

\funding{This work has been funded by the Deutsche
Forschungsgemeinschaft (DFG, German Research Foundation) – 390685813;
459420781 and by the Lamarr Institute for Machine Learning and
Artificial Intelligence lamarr-institute.org.}

\begin{document}
\maketitle

\begin{abstract}
The problem of computing the set of Pareto-optimal solutions has been studied for a variety of multiobjective optimization problems. For many such problems, algorithms are known that compute the Pareto set in (weak) output-polynomial time. These algorithms are often based on dynamic programming and by weak output-polynomial time, we mean that the running time depends polynomially on the size of the Pareto set but also on the sizes of the Pareto sets of the subproblems that occur in the dynamic program. For some problems, like the multiobjective minimum spanning tree problem, such algorithms are not known to exist and for other problems, like multiobjective versions of many NP-hard problems, such algorithms cannot exist, unless $\mathcal{P}=\mathcal{NP}$.

Dynamic programming over tree decompositions is a common technique in parameterized algorithms. 
In this paper, we study whether this technique can also be applied to compute Pareto sets of multiobjective optimization problems. 
We first derive an algorithm to compute the Pareto set for the multicriteria $s$-$t$ cut problem and show how this result can be applied to a polygon aggregation problem arising in cartography that has recently been introduced by Rottmann et al.~(GIScience 2021). 
We also show how to apply these techniques to also compute the Pareto set of the multiobjective minimum spanning tree problem and for the multiobjective TSP. 
The running time of our algorithms is $\bigO(f(\Tw)\cdot\mathrm{poly}(n,p_{\max}))$, where $f$ is some function in the treewidth~$\Tw$, $n$ is the input size, and $p_{\max}$ is an upper bound on the size of the Pareto sets of the subproblems that occur in the dynamic program.
Finally, we present an experimental evaluation of computing Pareto sets on real-world instances of polygon aggregation problems.
For this matter we devised a task-specific data structure that allows for efficient storage and modification of large sets of Pareto-optimal solutions. 
Throughout the implementation process, we incorporated several improved strategies and heuristics that significantly reduced both runtime and memory usage, enabling us to solve instances with treewidth of up to $22$ within reasonable amount of time.
Moreover, we conducted a preprocessing study to compare different tree decompositions in terms of their estimated overall runtime.
\end{abstract}

\setcounter{page}{0}

\newpage


\section{Introduction}
Multiobjective optimization problems arise naturally in many contexts. When booking a train ticket, one might be interested in minimizing the travel time, the price, and the number of changes. Similarly when planning a new infrastructure, companies often have to find a compromise between reliability and costs, to name just two illustrative examples. For multiobjective optimization problems there is usually not a single solution that is optimal for all criteria simultaneously, and hence, one has to find a trade-off between the different criteria. Since it is not a priori clear which trade-off is desired, one often studies the set of \emph{Pareto-optimal solutions} (also known as \emph{Pareto set}), where a solution is Pareto-optimal if it is not dominated by any other solution, i.e., there does not exist another solution that is better in at least one criterion and at least as good in all other criteria. The reasoning behind this is that a dominated solution cannot be a reasonable compromise of the different criteria and it makes sense to restrict one's attention to only the Pareto-optimal solutions.

A lot of research in multiobjective optimization deals with algorithms for computing the Pareto set and with studying the size of this set for various optimization problems. If this set is not too large then it can be presented to a (human) decision maker who can select a solution. The Pareto set also has the important property that any monotone function that combines the different objectives into a single objective is optimized by a Pareto-optimal solution. This means, in order to optimize such a function, one could first generate the Pareto set and then pick the best solution from this set.

This approach is only reasonable if there are not too many Pareto-optimal solutions. While in the worst case the Pareto set can be of exponential size for almost all multiobjective optimization problems, it has been observed that the Pareto set is often small in practice (see, e.g.,~\cite{Ehrgott2005,Muller-HannemannW06}). It has also been shown in the probabilistic model of smoothed analysis that for many problems the expected size of the Pareto set is polynomial if the input is subject to a small amount of random noise~\cite{BrunschR15,BeierRRV2022}. 
This motivates the development of algorithms for computing the Pareto set that have polynomial running time with respect to the output size.

In the literature such output-sensitive algorithms have been developed for different multiobjective optimization problems (e.g., for the multiobjective shortest path problem~\cite{CorleyM85,Hansen80,SkriverA00}, multiobjective flow problems~\cite{Ehrgott99,MustafaG98}, and the knapsack problem when viewed as a bicriteria optimization problem~\cite{NemhauserU69}). There is, however, a small caveat. Since these algorithms are usually based on dynamic programming, they are not output-polynomial in the strict sense but they solve certain subproblems of the given instance and their running times depend not only polynomially on the size of the Pareto set of the entire instance but also polynomially on the sizes of the Pareto sets of the subproblems. While in theory it could be the case that some of the subproblems have Pareto sets of exponential size while the Pareto set of the entire instance is small (see, e.g.,~\cite{Boekler2018}), this behavior is neither observed in experiments nor does it occur in probabilistic input models. Hence, algorithms that are output-polynomial in the weak sense that their running time depends polynomially on the sizes of the Pareto sets of all subproblems that occur in the dynamic program are useful and state-of-the-art for many multiobjective optimization problems. For some problems such output-sensitive algorithms are not known to exist. In particular, for the multiobjective spanning tree problem no algorithm is known that computes the Pareto set in output-polynomial time (not even in the weak sense).

Dynamic programming is a common technique in parameterized algorithms. For many graph problems, dynamic programming on tree decompositions is applied to obtain fixed-parameter tractable (FPT) algorithms with respect to the treewidth of the input graph. In this paper, we explore for the first time the potential of dynamic programming on tree decompositions for computing Pareto sets.  
First we design an algorithm for the multicriteria $s$-$t$ cut problem and apply it on a problem from cartography that has recently been introduced by Rottmann et al.~\cite{RDHR21} and that was the original motivation for our research. 

Rottmann et al. study maps with building footprints and present a new model how less detailed maps can be derived from a given map. Their method is based on viewing the problem as a bicriteria optimization problem. It is assumed that the plane containing the building footprints is triangulated and a set of triangles to glue together some of the buildings is to be selected that on the one hand minimizes the total area and on the other hand minimizes the total perimeter. Rottmann et al. present an algorithm that computes the set of extreme nondominated solutions (i.e., the set of solutions that optimize a linear combination of the objectives), which is a subset of the Pareto set. 
They ask if it is also possible to compute the entire Pareto set in some output-efficient way. 

We show how a treewidth-based algorithm can be used to compute the Pareto set for the $s$-$t$ cut problem.
The running time of this algorithm is FPT in the treewidth and output-polynomial in the weak sense, i.e., it is of the form $\bigO(f(\Tw)\cdot \mathrm{poly}(n,\pmax))$ where $f$ is some function in the treewidth $\Tw$, $n$ denotes the input size, and $\pmax$ denotes an upper bound on the size of the Pareto sets of the subproblems that occur in the dynamic program.

As a special case, the results for the multiobjective $s$-$t$ cut algorithm directly translate to the cartography problem by computing the whole Pareto set in FPT time and being output-polynomial in the weak sense.
We experimentally analyze the computation of Pareto sets for this cartography problem on real-world datasets by designing a specialized data structure, developing heuristics, and integrating additional implementation-specific techniques tailored to the task.

\subsection{Our Results}
We first consider the multiobjective $s$-$t$ cut problem and its special case, the triangle aggregation problem~\cite{RDHR21}: Given some polygons $P$ and triangles $T$, the goal is to find all Pareto-optimal subsets $T'\subseteq T$ minimizing both the total area and perimeter of $P\cup T'$. 
This problem arises in cartography, where the goal is to compute a less detailed map from a given map.
Here the polygons are building footprints and the space between these footprints is triangulated.
By including triangles, one can glue together these building footprints, leading to a less detailed version of the map. 
We show that the multiobjective $s$-$t$ cut problem implies an algorithm to compute the Pareto set in time $\bigO(n\Tw\cdot 2^{\Tw} \cdot \pmax^2\log(\pmax))$ for this problem, where $\Tw$ is the treewidth of the triangle adjacency graph.

In the appendix, Section~\ref{apx:proofsMST} and~\ref{apx:proofsTSP}, we show that this method for computing Pareto sets can also be applied for other multiobjective optimization problems, such as the multiobjective minimum spanning tree problem (MST) and the multiobjective traveling salesman problem (TSP).
For both, we present algorithms with running time $\bigO(n\Tw^{\bigO(\Tw)}\cdot \pmax^2 \log^{\dd-2}(\Tw^{\mathcal{O}(\Tw)}\cdot \pmax^2))$ where $\Tw$ is the treewidth of the input graph, $n$ is the number of vertices in the graph, $\pmax$ is the size of the largest Pareto set computed in one of the subproblems of the dynamic program, and $\dd\geq 2$ is the (constant) number of different objectives. 
This shows that the Pareto set can be efficiently computed for graphs with small treewidth if the Pareto sets of the subproblems are not too large, which is often the case in realistic inputs.

Finally, we experimentally evaluate our algorithm on real-world datasets. 
We introduce new heuristic pruning techniques, runtime  estimates for tree decompositions and their root, memory optimization, multi-threading, and other improvements for reducing runtime and memory usage. 

\subsection{Related Work}

Output-polynomial time algorithms (at least in the weak sense) are known for many problems. Examples include the multiobjective shortest path problem~\cite{CorleyM85,Hansen80,SkriverA00}, multiobjective flow problems~\cite{Ehrgott99,MustafaG98}, and the knapsack problem when viewed as a bicriteria optimization problem~\cite{NemhauserU69}. For the multiobjective spanning tree problem such algorithms are not known and there are only results on the set of extreme nondominated solutions. This is a subset of the Pareto set and it contains all solutions that optimize some linear combination of the different objectives. For the multiobjective spanning tree problem it is known that there are only polynomially many extreme nondominated solutions~\cite{Dey97}, and efficient algorithms for enumerating them exist~\cite{AgarwalEGH98}.

Rottmann et al.~\cite{RDHR21} introduce the triangle aggregation problem under the objective of map simplification by grouping multiple building footprints together. They show that, similar to the multiobjective minimum spanning tree problem, there exist only a polynomial number of nondominated extreme solutions and they present an algorithm to compute this set in polynomial time. They also provide an extensive experimental study of their algorithms on real-world data sets showing that the model captures nicely the intention to aggregate building footprints to obtain less detailed maps. They put it as an open question whether or not it is possible to also compute the set of Pareto-optimal solutions in an output-efficient way.

Motivated by the observation that for many problems the Pareto set is often small in applications, the number of Pareto-optimal solutions has been studied in the probabilistic framework of smoothed analysis in a sequence of articles~\cite{RoglinT09,MoitraO12,BrunschR15,BeierRRV2022}. It has been shown for a large class of linear integer optimization problems (which contains in particular the multiobjective MST problem and the multiobjective TSP) that the expected number of Pareto-optimal solutions is polynomially bounded if the coefficients (in our case, the edge weights) are independently perturbed by random noise. Formally, the coefficients of $\dd-1$ out of the $\dd$ objectives are assumed to be independent random variables with adversarially chosen distributions with bounded density. It is shown that not only the expected number of Pareto-optimal solutions is bounded polynomially in this setting but also all constant moments, so in particular the expected squared number of Pareto-optimal solutions is also polynomially bounded. Combined with this, our results imply that the Pareto set for the multiobjective MST problem, the multiobjective TSP and the multiobjective $s$-$t$ cut problem can be computed in expected FPT running time $\bigO(f(\Tw)\cdot \text{poly}(n))$ in the model of smoothed analysis.

\begin{jkblock}
While dynamic programming over tree decompositions is a well-established technique, using it effectively in practice comes with more difficulties than simply minimizing the width of the decomposition. 
In experiments the runtime of the algorithm can vary largely though the width of the tree decomposition is the same in all cases.

Bodlaender and Fomin~\cite{BF05} introduced the concept of the \emph{$f$-cost} of a tree decomposition, which sums up a cost function $f$ applied to the bag sizes across all nodes. This framework formalizes the observation that not all nodes contribute equally to runtime or memory, and that practical efficiency can benefit from structurally favorable decompositions.

The $f$-cost approach was later extended and evaluated experimentally by Abseher et al.~\cite{ADMW15} using machine learning to select a tree decomposition based on different features such as bag sizes, branching factors, and node depths.
With their method they could often select a decomposition that performed well in practice.
Kangas et al.~\cite{KKS20} used the same model for the problem of counting linear extensions, confirming the value of such estimators, and suggesting the average join-node depth as a good single feature for estimating runtime.

Beyond selection heuristics, tree decompositions have also been studied with respect to memory efficiency. Charwat et al.~\cite{CW15} explored compressed representations of intermediate states using decision diagrams, and Betzler et al.~\cite{BRJ06} proposed anchor-based compression techniques to reduce hard drive memory usage.
 
In our case, we not only solve an optimization problem, but compute the entire set of Pareto-optimal solutions. This makes memory usage more sensitive to the structure of join operations and limits the applicability of pointer-based or diagram-compressed representations.
To address this, we develop a custom estimation strategy for predicting both runtime and memory usage, based on predicting the number of Pareto-optimal solutions at each node and the runtime impact of join operations.
We then select a decomposition based on a combination of both runtime and memory usage. 
Details of this procedure are explained further in Appendix~\ref{sec:appendix_choosing_a_td}.
\end{jkblock}

\section{Preliminaries} \label{sec:Preliminaries}

Let $[i]:= \{1,\ldots,i\}$.
We consider an arbitrary optimization problem with a constant number $\dd\in \mathbb{N}$ of objectives to be minimized. Each (feasible) solution~$s$ is mapped to a cost vector $f(s)=(f(s)_1,\ldots,f(s)_\dd) \in \mathbb{R}^\dd$ where $f$ is our objective function.
A solution $s_1$ is said to \emph{dominate} another solution $s_2$ if $f(s_1)_i \leq f(s_2)_i$ for all $i\in [\dd]$ and there exists some $j\in [\dd]$ such that $f(s_1)_j<f(s_2)_j$.
We write $s_1 \prec s_2$ to denote that $s_1$ dominates $s_2$.
The set of all non-dominated solutions is called the \emph{Pareto set}. 
These definitions extend naturally to optimization problems where a subset of objectives is to be maximized instead of minimized.  
In the following we assume that all objectives are to be minimized, but all arguments and results can be adapted analogously when some (or all) objectives are to be maximized.

Let $\mathcal{P}_1$ and $\mathcal{P}_2$ be two sets of Pareto-optimal solutions.
We define their combined Pareto set as $\PP_1 + \PP_2 := \{p \in \PP_1 \cup \PP_2 \mid p \text{ is Pareto-optimal in } \PP_1 \cup \PP_2\}$, their
union as $S^{\mathcal{P}_1,\mathcal{P}_2} := \{p_1\cup p_2 \mid p_1 \in \PP_1, p_2 \in \PP_2 \}$, which contains all pairwise unions of solutions from $\PP_1$ and $\PP_2$, and the Pareto set of $S^{\mathcal{P}_1,\mathcal{P}_2}$ as $\PP_1 \oplus \PP_2 := \{s\in S^{\mathcal{P}_1,\mathcal{P}_2} \mid  s\text{ is Pareto-optimal in }S^{\mathcal{P}_1,\mathcal{P}_2} \}$.

\subsection{Treewidth}
Our algorithms rely on dynamic programming over a \textit{tree decomposition}.
The concept of tree decomposition and \textit{treewidth}, introduced by Robertson and Seymour \cite{RS86}, provide a measure of how similar a graph is to a tree. 
Many \NP{}-hard problems become tractable on graphs with small treewidth.
Computing the optimal treewidth is \NP{}-hard \cite{ACP87}, but it is FPT with respect to the treewidth~\cite{B93}, \jk{and a tree decomposition of width at most $2\cdot \Tw(G)$ can be computed in linear time using approximation algorithms~\cite{KO21}.
}
We now briefly introduce the notation used for tree decompositions.

\begin{definition}[Tree decomposition]
A tree decomposition of a graph $G=(V,E)$ is a pair $\mathcal{T}=(\mathbb{T}=(\{t_1,\ldots,t_m\},E),\{X_t\}_{t\in V(\mathbb{T})})$ where $\mathbb{T}$ is a tree and each node $t\in V(\mathbb{T})$ is associated with a bag $X_t\subseteq V$, such that the following conditions hold:
\begin{itemize}
\item $\bigcup_{t\in V(\mathbb{T})}X_t = V(G)$, i.e., every vertex of $G$ appears in at least one bag,
\item $\forall \{u,v\}\in E(G)~\exists t\in V(\mathbb{T}): u\in X_t \wedge v \in X_t$, i.e., for every edge $\{u,v\}\in E(G)$ the vertices $u$ and $v$ must appear together in at least one bag,
\item $T_u = \{t\in V(\mathbb{T}) \mid u \in X_t\}$ is a connected subtree of $\mathbb{T}$ for every $u\in V(G)$.
\end{itemize}
\end{definition}
The \textit{width} of a tree decomposition is $w(\mathcal{T}) := \max_{t\in V(\mathbb{T})}\vert X_t\vert - 1$, and the \textit{treewidth} $\Tw(G)$ is the smallest possible \textit{width} of a tree decomposition.
We denote by $\Tw$ the treewidth of the input graph.
Additionally, we denote with $V_t$ the union of vertices that were introduced at $t$ or some child of $t$.
In the following, we write $\Tw$ for the width of the tree decomposition used by our algorithm, which may be larger than $\Tw(G)$ if an approximation is used.

We assume we have a \textit{nice tree decomposition}, which is a binary tree decomposition rooted at some node $r\in V(\mathbb{T})$ with $X_r=\emptyset$ and with four specific node types: 
\begin{itemize}
\item \textbf{Leaf node}: Let $t$ be a node with no child nodes. 
Then $t$ is a \textit{Leaf node} iff $X_t = \emptyset$.
\item \textbf{Introduce node}: Let $t_\text{parent},t_\text{child}$ be two nodes in $V(\mathbb{T})$, where $t_\text{parent}$ has exactly one child $t_\text{child}$. 
Then $t_\text{parent}$ is an \textit{Introduce node} iff $X_{t_\text{parent}} = X_{t_\text{child}}\cup \{v\}$ for some $v\in V(G)$.
\item \textbf{Forget node}: Let $t_\text{parent},t_\text{child}$ be two nodes in $V(\mathbb{T})$, where $t_\text{parent}$ has exactly one child $t_\text{child}$. 
Then $t_\text{parent}$ is a \textit{Forget node} iff $X_{t_\text{parent}} \cup \{v\} = X_{t_\text{child}}$ for some $v\in V(G)$.
\item \textbf{Join node}: Let $t_\text{parent},t_\text{child}^1, t_\text{child}^2$ be three nodes in $V(\mathbb{T})$, where $t_\text{parent}$ has exactly two children $t_\text{child}^1$ and $t_\text{child}^2$. 
Then $t_\text{parent}$ is a \textit{Join node} iff $X_{t_\text{parent}} = X_{t_\text{child}^1} = X_{t_\text{child}^2}$.
\end{itemize}

A nice tree decomposition of width $\Tw$ and $\bigO(\vert V(G) \vert \cdot \Tw)$ nodes can be computed in polynomial time from any tree decomposition with the same width~\cite{CF15}.

The algorithm works as follows: Given an instance $G=(V,E)$ and a cost function $\co:E\rightarrow \RR^\dd$, it first computes a nice tree decomposition $(\TT, \{X_t\}_{t\in V(\TT)})$. 
In every leaf node the Pareto sets are initialized as empty sets. 
Then, for each node $t$, depending on its type, the algorithm recursively computes the set of Pareto-optimal solutions for every subproblem at this node. 
We distinguish between the functions $\leafNode$, $\intrNode$ and $\forgNode$, corresponding to their respective node types and show correctness, assuming that the Pareto sets for all child nodes are correctly computed.


\section{\texorpdfstring{Multiobjective $s$-$t$ cut problem}{Multiobjective s-t cut problem}}\label{sec:M_st_C}
We consider the problem of computing all Pareto-optimal $s$-$t$ cuts in a graph. 
Given a graph $G'=(V\cup \{s,t\}, E')$ \jk{with $n=\vert V\vert$} and a cost function $\co':E'\rightarrow \mathbb{R}^\dd$, a cut corresponds to a subset $S \subseteq V$.
We also refer to $S$ as a \textit{selection}.
The cost of a cut $S$ is defined as the sum of costs of all edges crossing from $S \cup \{s\}$ to $(V \setminus S) \cup \{t\}$.
Let $\delta(S):= \{\{u,v\}\in E' \mid u\in S\cup \{s\}, v\in (V\setminus S) \cup \{t\} \}$ be the set of edges cut by $S$.
The cost of $S$ is given by $\sum_{e \in \delta(S)} \co'(e)$.
We now wish to compute all $s$-$t$ cuts that are Pareto-optimal. 
Without loss of generality, we assume $\{s,t\}\notin E'$, since such an edge only introduces a constant shift in cost to every Pareto-optimal solution.

We denote $\delta (A,B):=\{\{u,v\}\in E' \mid u\in A, v\in B\}$ for $A,B\subseteq V(G')$ as the set of edges between vertices in $A$ and $B$ that are being cut.
We define $\co:2^{V(G')}\times 2^{V(G')}\rightarrow \mathbb{R}^\dd$ with $\co(A,B)=\sum_{e\in \delta (A,B)}\co'(e)$ for any subsets $A,B\subseteq V(G')$.
For simplicity, for a single vertex $p\in V(G')$ and subset $A\subseteq V(G')$, we write $\co(p,A)$ instead of $\co(\{p\}, A)$. 

Let $G=G'[V]$ be the subgraph induced by $V$.
Given a nice tree decomposition $(\mathbb{T}, \{X_t\}_{t\in V(\mathbb{T)}})$ for $G$, we compute for each node $t\in V(\mathbb{T})$ and subset $S\subseteq X_t$ the Pareto set $\PP_t^S$ consisting of all Pareto-optimal solutions $p\subseteq V_t$ with $p\cap X_t = S$.
Every such solution $p$ has cost $\co'(\delta(p))=\co(p, V \setminus p)+ \co(p, t)+\co(V\setminus p, s)$ (vertex $v\in V\setminus V_t$ is associated with vertex $t$). 
We refer to such an instance by a pair $(t, S)$.

\begin{theorem}\label{thm:stCut}
For an arbitrary $s$-$t$ cut instance $(G'=(V\cup \{s,t\},E'),\co')$ with $\co':E\rightarrow \mathbb{R}^\dd$, we can compute the set of Pareto-optimal $s$-$t$ cuts in time $\bigO(n\Tw\cdot 2^{\Tw} \cdot \pmax^2 \log^{\max \{\dd-2, 1\}}(\pmax^2))$ if a nice tree decomposition of graph $G=G'[V]$ with width $\Tw$ is provided.
\end{theorem}
In the following we describe how to compute the Pareto set for each node type.
For each instance $(t,S)$, the set of Pareto-optimal solutions is stored in an entry $D[t,S]$.
For simplicity, we assume that $D[t,S]$ contains the actual corresponding Pareto set.
Since the algorithm only depends on the cost vectors and selection $S$, in the actual implementation we will only save their cost vectors and reconstruct each solution by additionally storing pointers to the solutions from which it originated.
Since we assume $\dd$ as a constant, the encoding length of every Pareto-optimal solution is constant as well.
The proofs of the following lemmas are deferred to~\cref{apx:proof_st_C}.

\smallskip
\noindent
$\leafNode(t,S)$: For a leaf node $t$, we have only one solution $D[t,\emptyset]=\{\emptyset\}$.

\smallskip
\noindent
$\intrNode(t,S)$: Assume we have $X_t = X_{t'} \cup \{v\}$. 
If $v\notin S$, then $D[t,S]=D[t',S]$.
If $v \in S$, then $v$ is only adjacent to vertices in $X_{t'}$ and $V\setminus V_t$. 
If $X_{t'}$ is fixed, $v$ being associated with $s$ only changes the cost of every solution by a fixed amount.
Therefore we have $D[t, S] =D[t',S\setminus \{v\}]$, and their costs change by 
$\co(v, X_{t'} \cup (V\setminus V_t) \cup \{t\}) - 2\co(S\setminus \{v\},v) - \co(s,v)$.
\begin{restatable}{lem}{IntroduceNodesLemmaST}
If node $t$ introduces some vertex $v$ and has a child $t'$ for which $D[t',S]$ has been computed, then $\intrNode(t,S)$ computes $D[t,S]$ in time $\bigO(\pmax)$ for any partition $S$.
\end{restatable}

\smallskip
\noindent
$\forgNode(t,S)$: Assume we have $X_t = X_{t'} \setminus \{v\}$. 
To obtain $D[t,S]$, we compute the union of the sets $D[t',S]$ and $D[t', S\cup \{v\}]$ and then remove all dominated solutions.

\begin{restatable}{lem}{ForgetNodesLemmaST}
If node $t$ forgets some vertex $v$ and has a child $t'$ for which all possible sets $D[t', S']$ have been computed, then $\forgNode(t,S)$ computes $D[t,S]$ in time $\bigO(\pmax \log^{\dd-2}(\pmax))$.
\end{restatable}

For the correctness of $\joinNode(t,S)$, the following lemma will be useful:

\begin{restatable}{lem}{LocalExchangeCuts}\label{lemma:LocalExchangeCuts}
For each Pareto-optimal solution $p\subseteq V$ and any node $t\in V(\mathbb{T})$, the solution $p\cap V_t$ is Pareto-optimal in the instance $(t, X_t\cap p)$.
\end{restatable}

\smallskip
\noindent
$\joinNode(t,\mathcal{S})$:
Assume we have $X_t = X_{t_1} = X_{t_2}$. 
Since $X_t$ separates $V_{t_1}\setminus X_t, V_{t_2}\setminus X_t$ and $V\setminus V_t$, we can combine every solution from $D[t_1, S]$ with every solution from $D[t_2, S]$ and then remove all dominated ones.

\begin{restatable}{lem}{JoinNodesLemmaST}
If node $t\in V(\mathbb{T})$ is a join node with children $t_1,t_2\in V(\mathbb{T})$ and sets $D[t_1, S]$ and $D[t_2, S]$ have been correctly computed, then $\joinNode(t,S)$ computes $D[t,S]$ in time 
$\bigO(\pmax^2 \log^{\max \{\dd-2, 1\}}(\pmax^2))$.
\end{restatable}

\begin{proof}[Proof of Theorem \ref{thm:stCut}]
The correctness follows directly from the previous lemmas for the different node types, and from $D[r,\emptyset]$ being the Pareto set for the entire instance.

The runtime is dominated by the join nodes. 
Our nice tree decomposition contains $\bigO(n\Tw)$ nodes, and for each node $t$ the number of subsets $S\subseteq X_t$ is bounded by $2^{\Tw+1}$. 
Computing $D[t,S]$ at a join node takes time $\bigO(\pmax^2 \log^{\max \{\dd-2, 1\}}(\pmax^2))$. 
Hence, the total running time is bounded by $\bigO(n\Tw\cdot 2^{\Tw} \cdot \pmax^2 \log^{\max \{\dd-2, 1\}}(\pmax^2))$.
\end{proof}

The technique of computing the Pareto set for the $s$-$t$ cut problem can also be applied to other multiobjective optimization problems as well. 
To illustrate this, we show how to compute the Pareto set for the multiobjective MST problem and the multiobjective TSP problem.
The proofs of the following theorems are deferred to the appendix, Section~\ref{apx:proofsMST} and~\ref{apx:proofsTSP}.

\begin{restatable}{theorem}{MMST}\label{thm:mmst}
For an arbitrary multiobjective spanning tree instance $(G=(V,E),\co)$ with $\dd\geq 2$ objectives and $\co:E\rightarrow \mathbb{R}^\dd$, we can compute the set of Pareto-optimal spanning trees in time $\bigO(n(2\Tw)^{\bigO(\Tw)}\cdot \pmax^2 \log^{\max\{\dd-2,1\}}((2\Tw)^{\bigO(\Tw)}\cdot \pmax^2))$, if a nice tree decomposition of graph $G$ with width $\Tw$ is provided.
\end{restatable}

\begin{restatable}{theorem}{MTSP}\label{thm:mtsp}
For an arbitrary multiobjective traveling salesman instance $(G=(V,E),\co)$ with $\dd\geq 2$ objectives and $\co:E\rightarrow \mathbb{R}^d$, we can compute the set of all Pareto-optimal tours in time $\bigO (n(3\Tw+3)^{\bigO(\Tw)} \cdot \pmax^2 \cdot \log^{\max\{\dd-2,1\}}((3\Tw+3)^{\bigO(\Tw)} \cdot \pmax^2))$, if a nice tree decomposition of graph $G$ with width $\Tw$ is provided.
\end{restatable}

\section{Experiments}\label{sec:Exp}
In this section we present the experimental evaluation of our proposed algorithm for a special case of the multiobjective $s$-$t$ cut problem: the bicriteria polygon aggregation problem, introduced by Rottmann et al.~\cite{RDHR21}.
An instance $(T,P)$ consists of a set $P = \{p_1,\ldots,p_m\}$ of polygons and a set $T = \{t_1,\ldots,t_n\}$ of triangles. 
For any polygon $s\in T\cup P$, let $A(s)>0$ denote its area and $P(s)>0$ its perimeter.
Given a set of triangles $S\subseteq T$, we define $A(S)$ as the total area of $S\cup P$ and $P(S)$ as the total perimeter of $S\cup P$. 
The goal is to compute the Pareto set $\mathcal{P}$ of  subsets $S\subseteq T$ minimizing both $A(S)$ and $P(S)$. 

As with many multiobjective optimization problems, the set of Pareto-optimal solutions for the bicriteria polygon aggregation problem can be of exponential size. 
In fact, one can more generally show that the bicriteria polygon aggregation problem can be seen as a generalization of the bicriteria knapsack problem.
The statement on the size of the Pareto set follows as a corollary.
The proof of the following theorem is deferred to \cref{apx:proofKP}.

\begin{restatable}{theorem}{KPTA}
For every instance $I$ of the bicriteria knapsack problem, there exists an instance $I'$ of the bicriteria triangle aggregation problem such that there is a one-to-one correspondence between the Pareto-optimal solutions in $I$ and $I'$. 
\end{restatable}

Our implementation follows the methodology outlined in Section~\ref{sec:M_st_C}, is written in Java, and was evaluated on the datasets used by Rottmann et al.\cite{RDHR21}, which consist of triangle aggregation problems derived from various cities and villages in Germany.
The source code is available at \url{https://github.com/Tarek-pub/Bicriteria_Aggregation}.

We use the $s$-$t$ cut construction from Rottmann et al.~\cite{RDHR21} to generating each instance. 
Additionally, we remove vertices corresponding to polygons from the graph~$G$, resulting in a more compact, yet equivalent graph structure (see Appendix~\ref{appendix:graph_simplification} for more details).
To reduce complexity, weights (i.e. area and perimeter) are rounded to one decimal place (perimeter in decimals), and only unique-weight solutions are retained.
Tree decompositions are computed using the JDrasil framework~\cite{jdrasil}.
Several observations during development led to significant improvements in runtime and memory usage, which we integrated iteratively into the implementation and quantify through an analysis on the Ahrem dataset (\cref{tab:algorithm_improvements}).
Experiments were run using 16 threads of an Intel Core i9-13900 with 100GB RAM. 
Multiple tree decompositions were generated for exactly 1 hour using 16 threads.

A detailed description of the implementation is provided in \cref{sec:appendix_extended_implementation}.
\cref{sec:appendix_experiments} contains a complete list of datasets used in the final evaluation, along with their respective runtimes and memory usage. 

\begin{table}
    \centering
    \begin{tabular}{|l|r|r|}
        \hline
        \textbf{Algorithm improvement} & \textbf{Runtime [h]} & \textbf{Storage usage [GiB]} \\
        \hline
        Outsourcing \& Pruning** & 13767.8* & 3707 \\
        Choosing a good root** & 8571.6* (-38\%) & 3991 (+8\%) \\
        Choosing from multiple tree decompositions & 7129.6* (-17\%) & 3647 (-9\%) \\
        Join node heuristic & 162.0 (-98\%) & 3647 (-0\%) \\
        Join-forget nodes & 8.2 (-95\%) & 1499 (-59\%) \\
        Introduce-join-forget nodes & 7.0 (-15\%) & 122 (-92\%) \\
        \hline
        \textbf{Summary} & \textbf{-99.95\%} & \textbf{-96.71\%} \\
        \hline
    \end{tabular}
    \caption{Performance comparison of different algorithmic improvements on the dataset Ahrem. 
    *: These values are extrapolated based on sampled subproblems and scaled proportionally to estimate full runtimes (see \cref{sec:appendix_extrapolation}). \\
    **: For the versions before choosing a good root or tree decomposition, we chose the median rated root / tree decomposition in this table for a fair comparison.}
    \label{tab:algorithm_improvements}
\end{table}

As noted in \cref{sec:M_st_C}, it suffices to store only the overall cost of a solution, rather than the full set of triangles included.
The actual solutions can be reconstructed later by maintaining, additionally to the cost, pointers to the solutions from which the current one originated. 

We briefly describe the data structure used in our implementation, as well as a rough description of the algorithm adaptations. 
Further implementation details and specific improvements are provided in \cref{sec:appendix_extended_implementation}.
Each solution $p$ is represented as a weighted pair $(A(p),P(p))$ for area and perimeter.

\subsection{Outsourcing and Pruning}
\label{sec:implementation_pruning}
Initially, the implementation was unable to solve even small datasets such as Osterloh (treewidth $14$) due to RAM exhaustion reaching up to 100GB. 
To solve this problem, we implemented an efficient outsourcing strategy that stores all currently unneeded solutions onto a hard drive, keeping only the necessary solutions in RAM. 
Each solution is stored as a 16-byte entry in a  dedicated \textit{origin-pointer} file.
Each DP table $D[t,S]$ is stored in a separate \textit{surface-pointer} file, containing solution IDs and their weights. 
These files are kept on disk and only loaded into RAM when needed.
This data structure allowed solving Osterloh without exceeding RAM limits, but other datasets remained unsolvable due to hard drive usage growing up to 2TB. 
To free up disk space, we remove obsolete surface-pointer files and, when necessary, run a slow in-place pruning step to clean the origin-pointer file. 
This involves recursively marking only reachable entries by following their pointers and compacting the file accordingly. 
Pruning is only triggered when space is critically low, based on conservative estimates of future storage use.
These strategies prevented further storage issues, but many datasets were still unsolvable due to excessive runtime, particularly in join nodes. 
For example, the runtime required to solve the Ahrem dataset at this stage was estimated at approximately 13,768 hours, or one and a half year worth of CPU.

\subsection{Choosing an appropriate tree decomposition}
When generating multiple tree decompositions for the same graph, we observed that both the decomposition structure and the choice of the root node can significantly affect performance.
To estimate these effects\jk{, similar to Abseher et al.~\cite{ADMW15}}, we implemented a preprocessing simulation:
Each decomposition is traversed in postorder, and we estimate the number of solutions at each node to predict runtime and memory usage. 
We then select the decomposition that minimizes a weighted combination of both.

Earlier versions of our approach focused solely on minimizing runtime, which did not always lead to reduced memory usage (see \cref{tab:algorithm_improvements}). 
Only in the final version did we optimize both runtime and memory simultaneously. 
A detailed breakdown of these improvements is provided in \cref{sec:appendix_choosing_a_td}.

Incorporating this performance estimator into the algorithm led to a significant improvement in runtime. 
For the Ahrem dataset we estimated a runtime decrease to around 7,130 hours, a $48\%$ decrease in the estimated runtime.

\jk{
Additionally, in~\cref{sec:appendix_choosing_a_td} we show that our estimator reliably selects an almost optimal tree decomposition across our datasets. 
In nearly all cases, the decomposition with the best predicted  score closely matched the actual best decomposition in terms of lowest runtime and memory consumption. 
For example, in an experiment with 100 decompositions, our estimator selected the decomposition that ranked first in both runtime and storage.
The correlation between the predicted score and actual performance was $0.91$ for runtime and $0.97$ for storage.
We obtained similarly strong results when estimating performance for different root node choices, again observing a high correlation between predicted and actual runtime and memory usage.
}

\subsection{Join node algorithm}  
\label{sec:join_heap}
Let \( t \) be a join node with children \( t_1 \) and \( t_2 \). 
For each subset \( S \subseteq X_t \), we compute \( \PP_t^S =\PP_{t_1} \oplus \PP_{t_2} \) by combining all solutions \( p_1 \in \PP_{t_1}^S \) with all solutions \( p_2 \in \PP_{t_2}^S \) and subsequently remove all dominated solutions. 
To do this efficiently, we employ a lexicographical ordering of the input lists and process the combinations using a min-heap. 
Assume \( |\PP_{t_1}^S| \leq |\PP_{t_2}^S| \). 
We initialize a min-heap with one heap node for each \( p_1 \in \PP_{t_1}^S \), each pointing to the first entry in \( \PP_{t_2}^S \). 
We then repeatedly extract the root of the min-heap, process the corresponding solution pair \( (p_1, p_2) \), and increment its pointer to reference the next solution \( p_2' \in \PP_{t_2}^S \) that follows \( p_2 \) in lexicographical order. 
If such a solution \( p_2' \) exists, we update the heap accordingly, otherwise, we remove the heap node from the data structure. 
This process continues until the heap is empty, ensuring that every valid solution pair has been considered.  
This min-heap mechanism ensures that all candidate pairs are enumerated in lexicographical order, which enables an efficient check for Pareto-optimality.
A pseudocode description can be found in \cref{apx:join_algorithm}.

\subsection{Optimizing join node computations}
As with many dynamic programming algorithms over a nice tree decomposition, join nodes are the main computational bottleneck. 
At each join node $t$ we needed to compute $\PP = \PP_1 \oplus \PP_2$ for two Pareto sets $\PP_1,\PP_2$ exactly $2^{\vert X_t\vert}$ many times.
Computing $\PP_1 \oplus \PP_2$ for $\dd\geq 2$ can be done in $\bigO(p_{\max}^2 \log^{\max\{\dd-2,1\}}(\pmax^2))$ (see \cref{apx:pareto_runtime}).
In the worst case, $\PP_1 \oplus \PP_2$ can contain up to $\vert \PP_1\vert \cdot \vert \PP_2 \vert$ many solutions if every combination is Pareto-optimal.
By considering every combination $(p_1,p_2)\in \PP_1\times \PP_2$ and filtering out all dominated ones, this becomes very time-consuming.
However, in practice we observed only a roughly linear growth in size relative to the larger input Pareto set (see \cref{apx:join_solutions}).
This discrepancy indicated that many combinations were computed that would ultimately be discarded.
To circumvent this problem, we utilize a heuristic   pruning strategy to efficiently exclude many non-Pareto-optimal combinations at the same time. 

We construct a heuristic set $\HH$ for $\PP$ using a small subset of $\PP_1$ and $\PP_2$.
Afterwards we partition $\PP_1,\PP_2$ and $\HH$ into multiple sections and compute lower bounds for each section in $\PP_1,\PP_2$ and upper bounds for $\HH$. 
We then iterate over all combinations of lower bounds and compare them with the upper bounds. 
If a combination of lower bounds is entirely dominated by all upper bounds, we can conclude that no solution in these sections will be Pareto-optimal (PO) in $\PP$ and safely skip them in the min-heap.
To construct $\HH$, we apply the same optimization recursively until a base case is reached. 
Furthermore, we apply multi-threading for each join node $t$, enabling parallel computation of multiple entries $D[t,S]$ at the same time.

The heuristic pruning process introduces a trade-off between its own runtime and the efficiency gains from skipping combinations.
By selecting appropriate parameters, we achieved a favorable balance, reducing the runtime by nearly $98\%$ and ultimately solving Ahrem in under seven days.

\subsection{Join-forget nodes}
To further improve runtime and reduce memory usage, we analyzed recurring patterns in the algorithm process that allowed for optimization.
In many instances, after computing a join node, many of these solutions were immediately discarded in subsequent forget nodes. 
To address this inefficiency, we introduce a new node type, called \emph{join-forget} node, which merges a join node with subsequent forget nodes into a single operation.
Formally, for a join node~$t$ in the tree decomposition, if its unique parent and all ancestors up to the next non-forget node are forget nodes, we replace~$t$ and all these forget nodes by a single join-forget node~$t'$. 
This new node retains the original bag $X_t$ and additionally stores the set $F$ of forgotten vertices from the removed forget nodes. 

We consider which lists would be merged in the upcoming forget nodes of the original join node while combining solutions in the min-heap. 
More specifically, for a join-forget node with forgotten vertices $F$, we only create a min-heap for each subset $S' \subseteq X_t \setminus F$. 
Each such heap allows us to efficiently iterate over all combinations from the union
$\bigcup_{S\in \{S' \cup F' \mid F'\subseteq F \}} \{p_1\cup p_2 \mid p_1 \in \mathcal{P}_{t_1}^S, p_2 \in \mathcal{P}_{t_2}^S \}$.

This setup also enhances our join node heuristic. 
Instead of computing a heuristic set $\HH$ and upper bounds for each pair $\mathcal{P}_{t_1}^S,\mathcal{P}_{t_2}^S$ separately, we reuse the join-forget structure.
Specifically, for each subset $S' \subseteq X_t \setminus F$, we compute the same heuristic set $\HH$ for all $S\in \{S' \cup F' \mid F'\subseteq F \}$. We do so by using subsets of $\mathcal{P}_{t_1}^S$ and $\mathcal{P}_{t_2}^S$ for all such $S$.
As before, we recursively apply this heuristic until a small base case is reached.
A side effect of this strategy is that it significantly increases RAM usage, as each thread now requires not just three lists in memory (namely $\mathcal{P}_{t_1}^S,\mathcal{P}_{t_2}^S$ and $\mathcal{P}_{t_1}^S \oplus \PP_{t_2}^S$), but up to $2^{\vert F \vert +1}+1$ lists.

This strategy led to another significant boost in runtime for many of our datasets. 
For the Ahrem dataset, using join-forget nodes reduced the runtime by $95\%$, allowing us to solve the instance in almost $8$ hours. 
Moreover, storage usage also decreased by $59\%$.
We emphasize that the improvements are not solely due to the use of join-forget nodes, but by adapting the tree decomposition performance estimator accordingly as well.

\begin{table}
    \centering
    \begin{tabular}{lrrrrrrrrr}
        \toprule
        \shortstack[t]{Dataset \\~} & \shortstack[t]{$\Tw$\\~} & \shortstack[t]{Time\\{[h]}} & \shortstack[t]{Storage\\{[GIB]}} & \shortstack[t]{\#PO\\Solutions} & \shortstack[t]{Percentage\\nonextreme} & \shortstack[t]{Graph\\\#Vertices} & \shortstack[t]{Graph\\\#Edges} & \shortstack[t]{$\vert V(\mathbb{T}) \vert$\\~} \\
        \midrule
        Osterloh   & 14 & 0.15 & 5.8   & 83055   & 99.43 & 1717 & 1887 & 7586   \\
        Ahrem      & 17 & 3.1  & 122   & 219969  & 98.99 & 5280 & 5607 & 21778  \\
        Lottbek    & 19 & 10.5 & 355   & 316567  & 99.48 & 5595 & 6101 & 24426  \\
        Erlenbach  & 22 & 70.7 & 1754  & 303565  & 99.47 & 5644 & 6284 & 25682  \\
        \bottomrule
    \end{tabular}
    \caption{Overview over some of the datasets we managed to solve. For all datasets see \cref{tab:all_datasets} in the appendix.}
    \label{tab:some_datasets}
\end{table}

\subsection{Introduce-join-forget nodes}
Together with the newly defined join-forget node, we searched for new patterns in the algorithm computations to reduce the storage consumption further. 
To this end, we introduce a second node type called \emph{introduce-join-forget} node, which further generalizes the concept of join-forget nodes. 
After replacing applicable structures with join-forget nodes, we consider each join-forget node $t$ with child nodes $t_1,t_2$, where at least one child ($t_1$, $t_2$, or both) is an introduce node. 
We remove all consecutive introduce nodes among the children and store their introduced vertices as additional information in $t$.

A major inefficiency of standard introduce nodes arises from the exponential growth in possible subsets $S \subseteq X_t$. 
When a new vertex $v$ is introduced, each existing solution $p \in \PP_t^S$ for $S \subseteq X_t \setminus \{v\}$ is, with a constant weight adjustment, duplicated for both $S$ and $S\cup \{v\}$. 
This significantly increases both memory usage and I/O overhead in the outsourcing step, making it a significant bottleneck.
To resolve this, we avoid explicit duplication of solutions in introduce-join-forget nodes. 
Let $I \subseteq X_t$ be the set of vertices whose introduction was skipped on one side. 
If a set $\PP^S$ is required for the min-heap computations and $S \cap I \neq \emptyset$, we instead load the solutions of the surface-pointer file of $\PP^{S \setminus I }$ and add the constant weight adjustments according to $S\cap I$.
Additionally, we create an origin-pointer entry for such a solution only if they are identified as Pareto-optimal in the min-heap. 
This strategy also reduces the growth of the origin-pointer file.

Together with adapting the performance estimator, this final improvement reduced the memory consumption by an additional $92\%$, resulting in a total of only 122GB of storage required for the Ahrem dataset, while also achieving a further $15\%$ decrease in runtime, bringing the total runtime down to 7 hours.

\begin{figure}[t]
    \centering
    \includegraphics[width=0.32\textwidth]{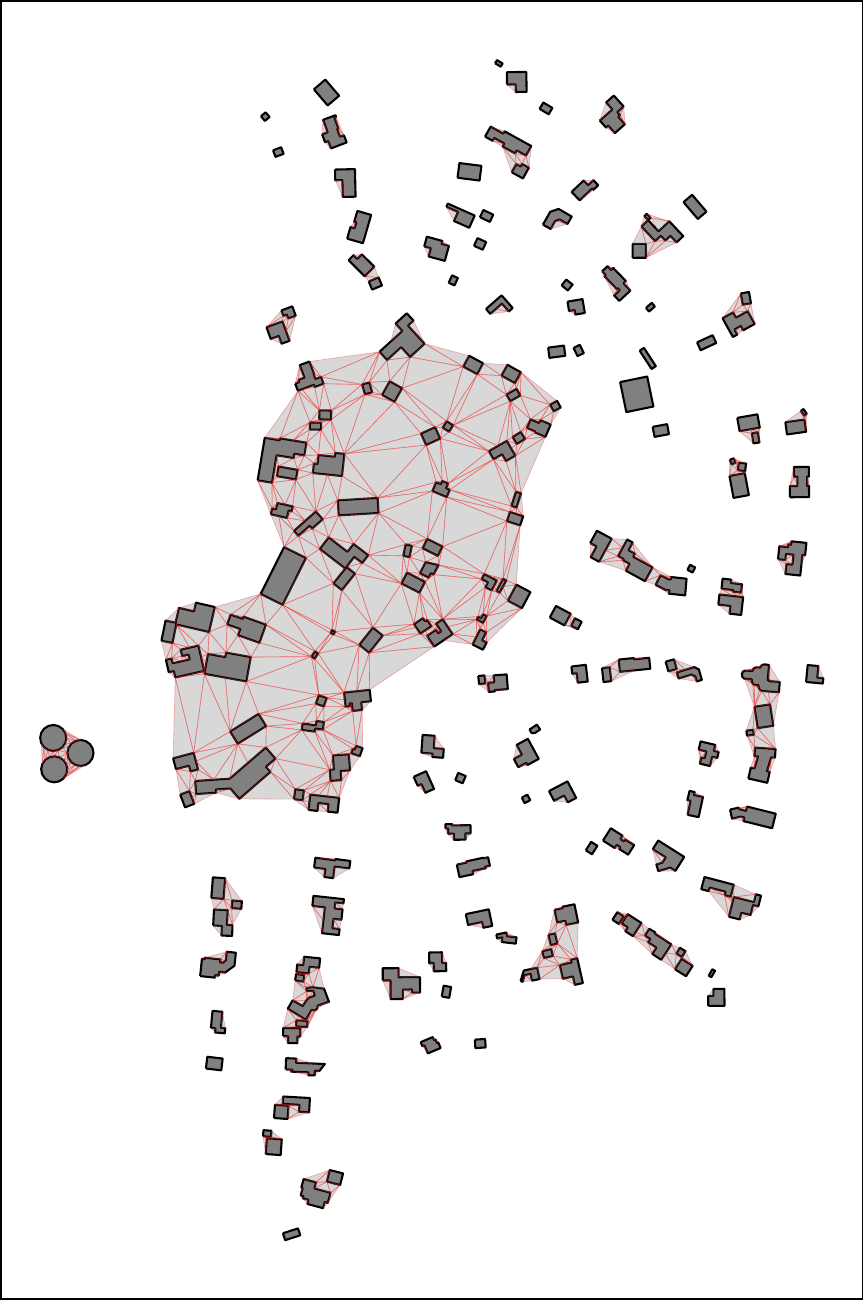}
    \includegraphics[width=0.32\textwidth]{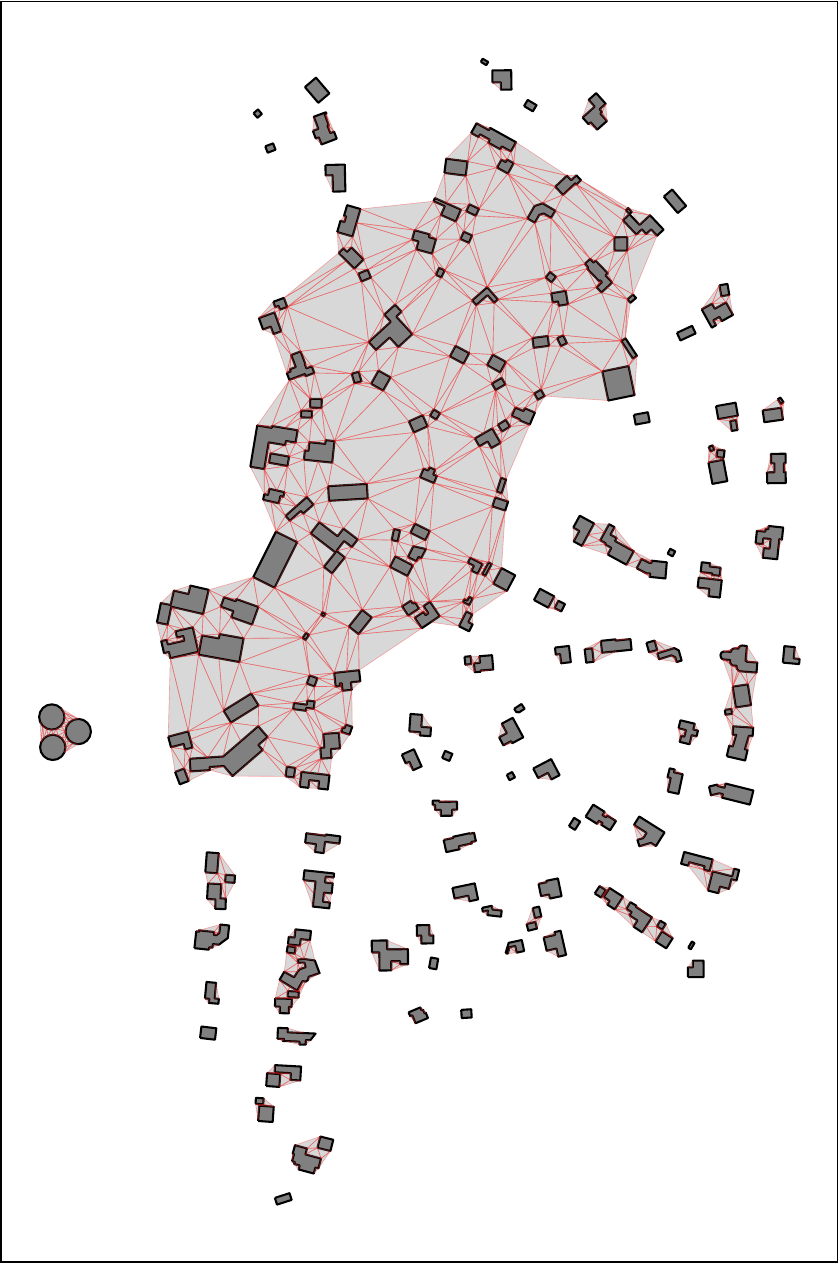}
    \includegraphics[width=0.32\textwidth]{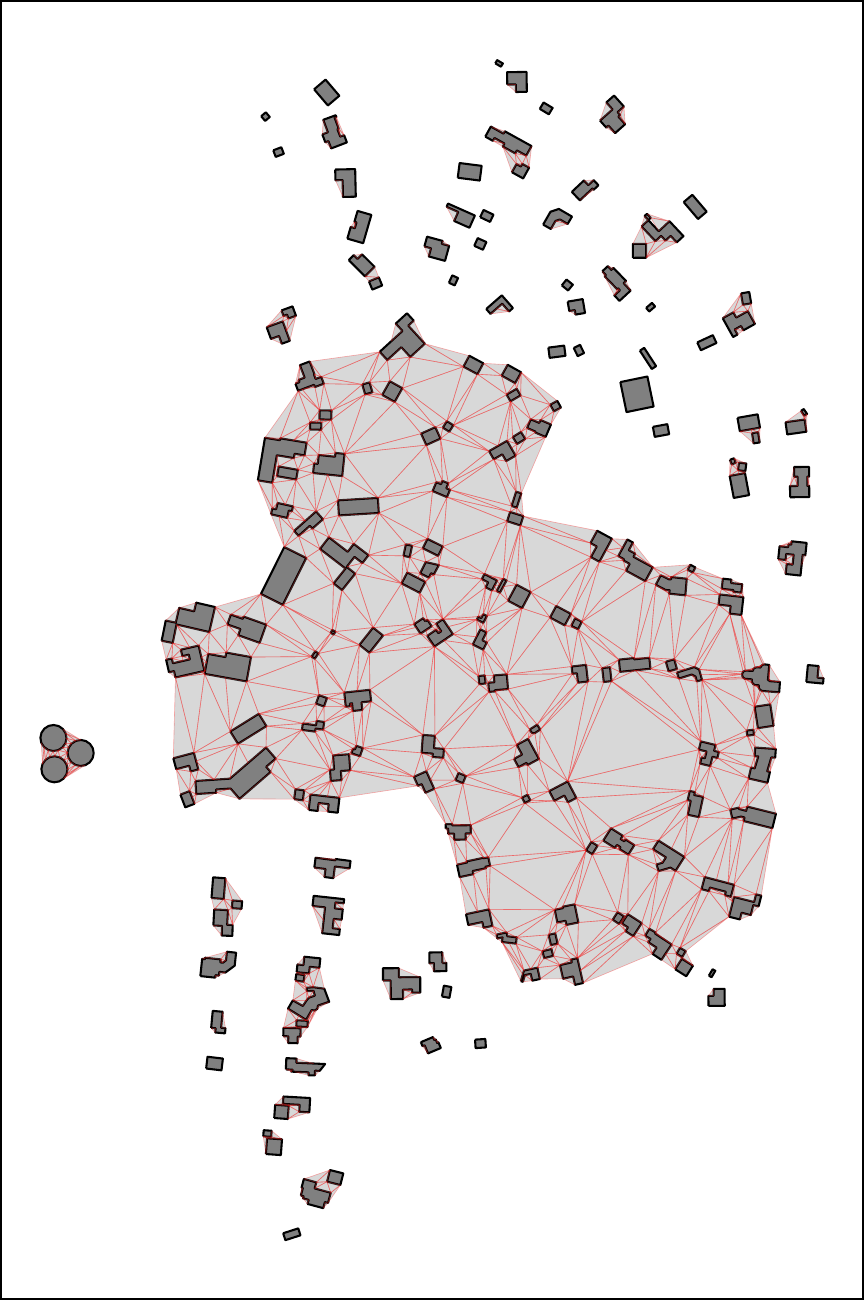}
    \caption{Three PO aggregations of the dataset Osterloh. Input polygons are filled dark gray, chosen triangles are filled light gray. Middle: nonextreme solution. Left/Right: Closest extreme solutions to the nonextreme solution.}
    \label{fig:extreme_vs_nonextreme}
\end{figure}

\section{Results}
Using an efficient data structure and multiple algorithmic improvements described in \cref{sec:Exp}, we were able to compute the full Pareto set for many datasets, including instances with treewidths of up to $22$, within a reasonable amount of time.
\cref{tab:some_datasets} shows a subset of datasets we successfully solved. 
All of these computations were performed on a high performance computing system. 
We used 96 threads of two Intel Xeon ``Sapphire Rapids'' 2.10GHz and about two terabytes of storage.
The full table and a more detailed description can be found in the appendix, \cref{tab:all_datasets}.

For the polygon aggregation problem, it is known that all extreme solutions are hierarchically compatible~\cite{RDHR21}. In \cref{fig:extreme_vs_nonextreme} we illustrate this by showing two consecutive extreme solutions $p_1,p_2$ and one nonextreme solution $p'$ that lies between them (i.e. $A(p_1)<A(p')<A(p_2)$ and $P(p_1)>P(p')>P(p_2)$). 
While the nonextreme solution $p'$ is also a viable aggregation, it exhibits a significantly different structure compared to the extreme solutions. 
Focusing solely on extreme solutions may lead to large gaps in the solution space, as seen in \cref{fig:osterloh_weights}.

Across the datasets we solved, extreme solutions accounted for only less than one percent of the total PO solutions on average (see \cref{tab:all_datasets}).
As a result, computing the set of nonextreme solutions can therefore lead to a significant richer range of solutions for the user to choose from, while the extreme solutions are forced being hierarchical depending on each other.

We published a tool to interactively investigate all PO solutions for the dataset Osterloh. 
The tool is available at \url{https://github.com/Tarek-pub/Bicriteria_Aggregation_plotting}.

\begin{figure}[t]
    \centering
    \includegraphics[width=0.5\textwidth]{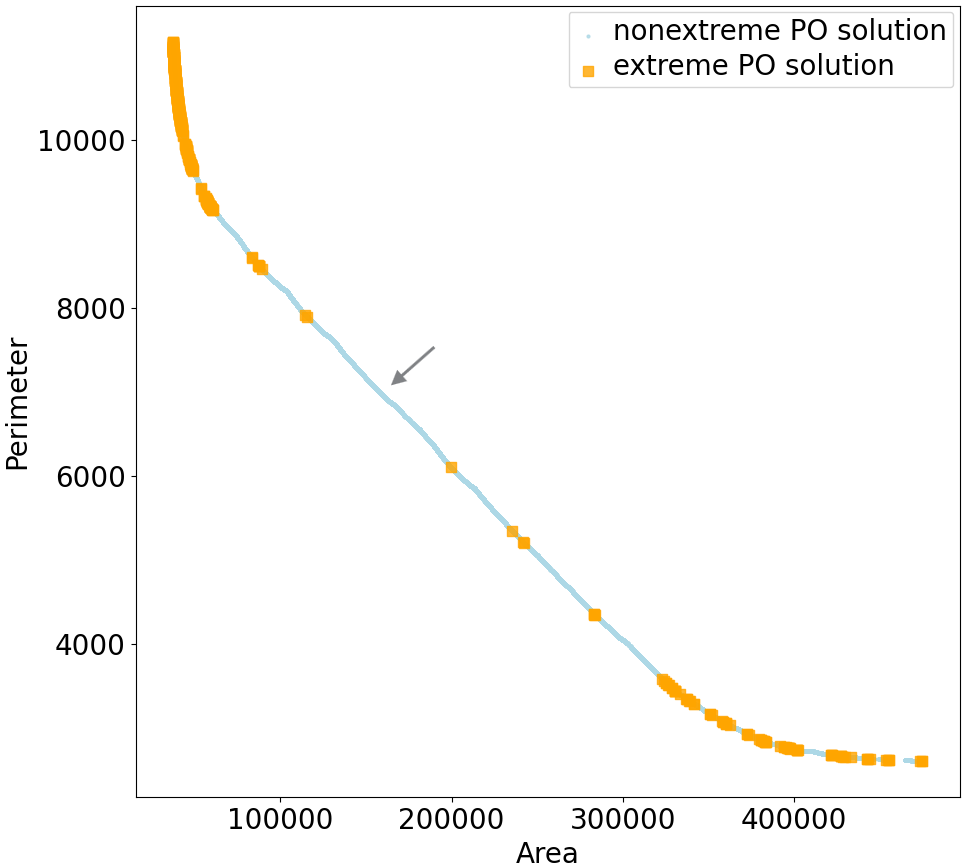}
    \caption{The weights of all PO solutions of Osterloh. The arrow indicates which nonextreme solution is shown in \cref{fig:extreme_vs_nonextreme}.}
    \label{fig:osterloh_weights}
\end{figure}
\section{Conclusions}

We presented the first algorithms for computing Pareto sets using dynamic programming over tree decompositions and showed that this framework can naturally be applied to various multiobjective optimization problems. 
The main motivation for our work was the article of Rottmann et al.~\cite{RDHR21}, who raised the question of whether it is possible to compute the Pareto set in output-polynomial time.  
We also conducted an experimental analysis of the polygon aggregation problem on real-world instances and developed several techniques to improve both runtime and memory usage.

We want to highlight that the theoretical running times of our algorithms are only worst-case bounds. 
In practice, because not every bag in a tree decomposition has the maximum possible size and the number of join nodes is often relatively small, the actual running time will usually be smaller. 
Additionally, it is a common phenomenon for multiobjective optimization that Pareto sets are not too large for real-world inputs. 
Hence, the dependency on the size of the largest Pareto set is not prohibitive in practice. 
For problems like the multiopjective minimum spanning tree problem, the multiopjective TSP, and many other problems, this is further supported by theoretical results in smoothed analysis, which shows that Pareto sets are expected to be polynomially bounded when instances are randomly perturbed~\cite{BrunschR15}. 
\bibliography{refs}
\newpage
\appendix

\section{Pareto sets computation}\label{apx:pareto_runtime}

In a naive approach, computing $\PP_1 + \PP_2$ takes time $\bigO(\vert \PP_1 \vert \vert \PP_2 \vert)$, and computing $\PP_1 \oplus \PP_2$ takes time $\bigO(\vert \mathcal{P}_1 \vert^2\vert \mathcal{P}_2\vert^2)$. 
This is because there are $\vert \mathcal{P}_1 \vert\vert \mathcal{P}_2\vert$ possible combinations of one solution from each set, and removing dominated solutions requires pairwise comparisons between all resulting combinations. Since the number $\dd$ of objectives is assumed to be constant, combining two solutions and comparing their cost vectors can be done in constant time. 
In this analysis we ignored the encoding length of a solution, which would contribute another factor into the running time. 
However, in our algorithms we do not store the solutions explicitly, but instead only their cost vectors, which is sufficient for our purposes and is of constant encoding length for each solution.

Assume the size of each Pareto set can never exceed some value $\pmax$. We can improve the running time of both operations as follows:
For $\PP_1 + \PP_2$ and $\dd=2$, if both sets are sorted by their first objective, we can compute the combined Pareto set via a simple sweep procedure in time $\bigO(\max\{\vert \PP_1 \vert, \vert \PP_2 \vert\}) = \bigO(\pmax)$.
For $\dd \geq 3$, to compute $\sum_{i=1}^m \PP_i$ we can maintain a pointer $p_i$ in each sorted list $\PP_i$ and store the current candidate in a min-heap, allowing retrieval the next-best solution in $\bigO(1)$ and updating the heap in $\bigO(\log m)$.
Since each pointer is incremented at most $\pmax$ times, we perform at most $\pmax \cdot \log m$ heap operations, yielding a total running time of $\bigO(\pmax \cdot m \log m)$. 
For $\dd\geq 3$ we use an algorithm introduced by Kung et al.~\cite{Kung1975}, which finds all componentwise maximal vectors in a $\dd$-dimensional set in $\bigO(n\log n)$ time for $\dd = 2$ and in $\bigO(n \log^{\dd-2}n)$ time for $\dd\geq 3$.
In case of $\dd\geq 3$, we can first compute the union $\PP_1 \cup \PP_2$ and then apply Kung et al.'s algorithm to extract the final Pareto set. This yields a running time of $\bigO((\vert \PP_1 \vert + \vert \PP_2 \vert) \cdot \log^{\dd-2}(\vert \PP_1 \vert + \vert \PP_2 \vert)) = \bigO(p_{\max} \cdot \log^{\dd-2}(\pmax))$.
Thus, for any $\dd\geq 2$, we obtain a running time of $\bigO(\pmax \cdot \log^{\dd-2}(\pmax))$ for computing $\PP_1 + \PP_2$. 
To compute the combined Pareto set $\sum_{i=1}^m \PP_i$ for $m\geq 3$ and $\dd\geq 3$, we merge all sets $\PP_i$ and then use the algorithm by Kung et al., resulting in a runtime of $\bigO(\pmax \cdot m \cdot \log^{\dd-2}(\pmax \cdot m))$.

For $\PP_1 \oplus \PP_2$, we first compute the set $S^{\mathcal{P}_1,\mathcal{P}_2}$ of size at most $\vert \mathcal{P}_1\vert\vert\mathcal{P}_2\vert \leq p_{\max}^2$ in time $\bigO(\vert \mathcal{P}_1\vert\vert\mathcal{P}_2\vert) = \bigO(p_{\max}^2)$, and then compute $\mathcal{P}_1\oplus \mathcal{P}_2$ from this set in \\$\bigO(\vert\mathcal{P}_1\vert\vert\mathcal{P}_2\vert \log^{\max\{\dd-2, 1\}}(\vert \mathcal{P}_1\vert \vert \mathcal{P}_2 \vert)) = \bigO(p_{\max}^2 \log^{\max\{\dd-2,1\}}(p_{\max}^2))$ for $\dd\ge 2$, using the algorithm by Kung et al.~\cite{Kung1975}. 
Therefore, for every $\dd\geq 2$ the final running time for computing $\PP_1 \oplus \PP_2$ is $\bigO(p_{\max}^2 \log^{\max\{\dd-2,1\}}(p_{\max}^2))$.


\section{Nice tree decomposition with introduce edges}\label{apx:introduce_edges}
To simplify the description for the multiobjective minimum spanning tree problem and multiobjective traveling salesman problem, we follow the approach in~\cite{CF15} and adjust the nice tree decomposition to include an additional type of node that introduces edges one by one. 
For every $t\in V(\mathbb{T})$, we define $V_t$ and $E_t$ as the set of vertices and edges that are revealed up to node $t$, respectively. The graph $G_t=(V_t,E_t)$ is then the subgraph of $G$ revealed up to node $t$. 
Formally, $V_t$ is the union of all $X_{t'}$ for nodes $t'\in V(\mathbb{T})$ in the subtree rooted at $t$. 
We define $E_t$ recursively as follows: $E_t=\emptyset$ for all leaves $t\in V(\mathbb{T})$. 
Now, let $t\in V(\mathbb{T})$ be an introduce or forget node, and let $t'\in V(\mathbb{T})$ be its child. Then we define $E_t=E_{t'}$. 
If $t\in V(\mathbb{T})$ is a join node and $t_1,t_2\in V(\mathbb{T})$ are its children, we define $E_t=E_{t_1}\cup E_{t_2}$. 
Edges are introduced by the following type of nodes.

\begin{itemize}
    \item \textbf{Introduce edge node}: For every edge $e=\{u,v\}\in E$ there is exactly one introduce edge node $t\in V(\mathbb{T})$ with $u,v\in X_t=X_{t'}$ for its only child bag $t'$ and $E_t = E_{t'}\cup \{e\}$.
\end{itemize}

For every edge $\{u,v\}\in E$, there must exist at least one node $t$ with $u,v\in X_t$. 
If we follow the unique path from $t$ to the root, we have to pass by a forget node $t'$, which forgets either $u$ or $v$. 
Without loss of generality, we assume that the introduce edge node for $\{u,v\}$ is placed between node $t'$ and its only child node. 
One consequence of this ordering is that for a join node $t\in V(\mathbb{T})$ the set $E_t$ does not contain any edges $\{u,v\}$ for $u,v\in X_t$, because these are introduced only further upwards in the tree.
When showing correctness for all possible types of nodes, we will denote the function that handles edge nodes as $\intrEdgeNode$.

It is worth noting that, given a tree decomposition of width $\Tw$ for some graph $G$, a nice tree decomposition of the same width $\Tw$, which also includes introduce edge nodes, and contains $\bigO(\vert V(G)\vert \cdot \Tw)$ nodes can also be computed in polynomial time~\cite{CF15}.

\section{\texorpdfstring{Proofs of Lemmas in multiobjective $s$-$t$ cut Section~\ref{sec:M_st_C}}{Proofs of Lemmas in multiobjective s-t cut Section~\ref{sec:M_st_C}}}\label{apx:proof_st_C}

\begin{figure}[t]
\centering
\includegraphics[width=.5\linewidth]{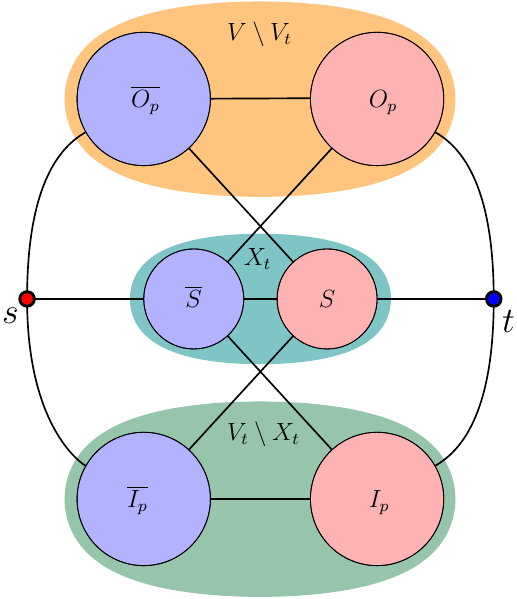}
\caption{Visualization of a partitioning for some solution $p\subseteq V$. All sets drawn in red are associated with vertex $s$ and all blue sets with $t$. Only drawn edges can contribute to the cost $w'(\delta(p))$.}
\label{fig:stCutNeighbourhood}
\end{figure}

\begin{figure}[t]
\begin{minipage}{.45\linewidth}
\includegraphics[width=\linewidth, page=2]{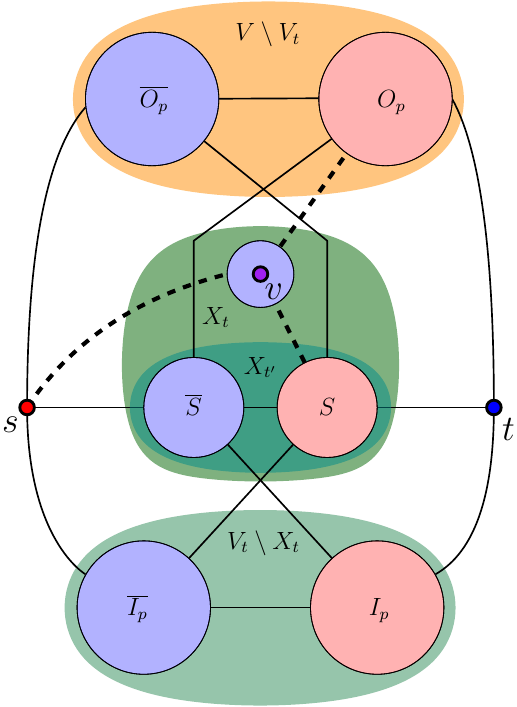}
\end{minipage}
\hfill
\begin{minipage}{.45\linewidth}
\includegraphics[width=\linewidth, page=3]{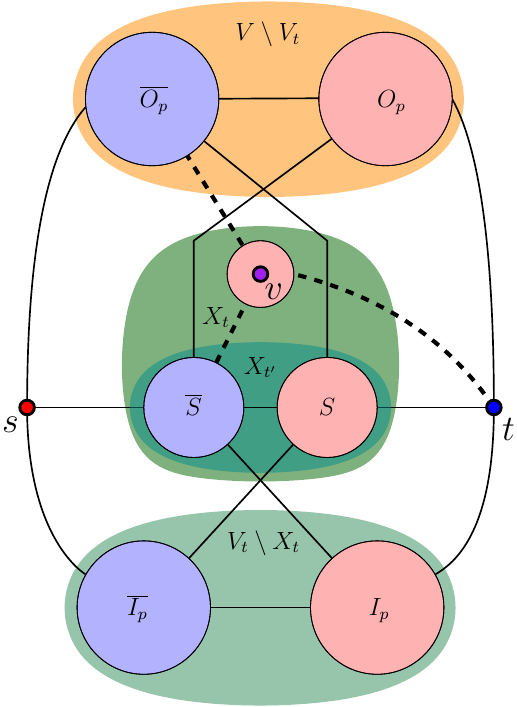}  
\end{minipage}
\caption{Transition of a vertex $v$ for an introduce node. Changing costs depending on $v$ are shown with dashed lines.}
\label{fig:stCutNeighborhoodIntro}
\end{figure}

For ease of notation, we subdivide the set $V$ with respect to an arbitrary node $t\in V(\mathbb{T})$ and solution $p \subseteq V$ with bag intersection $S:= X_t\cap p$ into the following subsets:
\begin{itemize}
\item $I_p:= p\cap (V_t \setminus X_t)$: vertices assigned to $s$ in $V_t\setminus X_t$.
\item $\overline{I_p}:= (V_t \setminus X_t) \setminus p$: vertices assigned to $t$ in $V_t \setminus X_t$.
\item $\overline{S}:= X_t\setminus S$: vertices assigned to $t$ in $X_t$.
\item $O_p:= p \setminus V_t$: remaining vertices assigned to $s$.
\item $\overline{O_p}:= (V\setminus V_t)\setminus p$: remaining vertices assigned to $t$.
\end{itemize}
For a visualization of these sets, see Figure~\ref{fig:stCutNeighbourhood}.
Due to the properties of a tree decomposition, there cannot be any edges between $V\setminus V_t$ and $V_t\setminus X_t$.
Hence, we always have $\co(\overline{O_p}, I_p) = \co(O_p, \overline{I_p})=0$.

In the following, we assume that each table entry $D[t,\SSS]$ stores the full set of solutions and not only their respective costs. 
Afterwards, we describe how one can reconstruct the actual solutions, if only the costs and a pointer (or two pointers in the case of join-nodes) are given. 
Since each subproblem depends only on the cost vector and bag assignment $S$, it suffices to compare cost vectors to determine the solutions that are Pareto-optimal.

\IntroduceNodesLemmaST*

\begin{proof}
Let $p\subseteq V_t$ be a Pareto-optimal solution in $\PP_t^S$ for the instance $(t,S)$. 
We distinguish two cases, depending on whether the introduced vertex $v$ is included in $p$ or not.

If $v\notin p$, this solution must also be Pareto-optimal for the instance $(t',S)$.
Assume for a contradiction that there exists a  solution $q\in \PP_{t'}^{S}$ with $q \prec p$. 
Since $q$ is also feasible for the instance $(t, S)$, this contradicts the assumption that $p$ is Pareto-optimal in $(t,S)$.
By the assumption of optimality for all child entries, $D[t',S]$ must contain $p$, so $D[t, S]$ will also contain $p$.

\noindent Assume now $v \in p$ and let $\tilde{p}=p\setminus \{v\}$.
If $\tilde{p}$ is Pareto-optimal for the instance $(t',S \setminus \{v\})$, we will also compute $p$, since we have already computed all Pareto-optimal solutions $\PP_{t'}^{S\setminus \{v\}}$ for the instance $(t', S\setminus \{v\})$ and we consider every possibility to add $v$ to every solution in this set.

By computing the cost of $\tilde{p}$ we get 
\begin{align*}
\co'(\delta(\tilde{p}))&=\co(I_{\tilde{p}}, {\overline{I_{\tilde{p}}}}) + \co(S, \overline{S}) + \co(O_{\tilde{p}}, \overline{O_{\tilde{p}}}) \\
&+ \co(S, \overline{I_{\tilde{p}}}) + \co(I_{\tilde{p}}, \overline{S}) + \co(S, \overline{O_{\tilde{p}}}) + \co(O_{\tilde{p}}, \overline{S}) \\
& + \co(s, \overline{O_{\tilde{p}}}) + \co(s, \overline{S}) + \co(s, \overline{I_{\tilde{p}}}) \\
& + \co(O_{\tilde{p}}, t) + \co(S, t) + \co(I_{\tilde{p}}, t)
\end{align*}
where we used that there are no edges between the sets $V\setminus V_t$ and $V_t\setminus X_t$.

\noindent If we now consider the change in cost when adding $v$ to $S$, i.e., the difference of $\co'(p)$ and $\co'(\tilde{p})$, we obtain
\begin{align*}
\co'(\delta(p)) - \co'(\delta(\tilde{p}))  &= \co(v, (V\setminus p) \cup \{t\}) - \co(\tilde{p}\cup \{s\}, v)  \\
&= \co(v, \overline{S}\setminus \{v\}) + \co(v, \overline{O_{\tilde{p}}}) + \co(v, t)\\ 
&-\co(S\setminus \{v\}, v) - \co(O_{\tilde{p}}, v) - \co(s, v) \\
&= \co(v, \overline{S}\setminus \{v\}) + \co(v,V\setminus V_t) + \co(v,t) - \co(S\setminus \{v\},v) - \co(s,v)\\
&= \co(v, \overline{S}) + \co(v,S\setminus \{v\}) + \co(v, V\setminus V_t) + \co(v,t) \\
&-\co(S \setminus \{v\},v) - \co(S \setminus \{v\},v) - \co(s,v)\\
&=\co(v, X_{t'} \cup (V\setminus V_t) \cup \{t\}) - 2\co(S\setminus \{v\},v) - \co(s,v).
\end{align*}
This holds because, when $v$ is introduced, it belongs to the set $V\setminus V_{t'}$, and thus can only be adjacent to vertices from $X_{t'}$ or $V\setminus V_{t'}$, for a visualization, see Figure~\ref{fig:stCutNeighborhoodIntro}. 
We also used that $O_{\tilde{p}}=\emptyset$ and $\overline{O_{\tilde{p}}}=V\setminus V_t$ and $\co(v,X_{t'})=\co(v,\overline{S}) + \co(v,S\setminus \{v\})$.
If the subset $S\subseteq X_{t'}$ is fixed and we add $v$ to any solution $p\subseteq V_{t'}\setminus X_{t'}$, then the cost $\co'(p)$ only changes by a constant amount that is independent of $p$.
Therefore, if a solution $p\subseteq V_{t}\setminus X_{t}$ is Pareto-optimal in the instance $(t, S)$, its remaining subsolution $p\setminus \{v\}$ will also be Pareto-optimal for the instance $(t', S \setminus \{v\})$.

For the computation of the final set we only need to add for each solution $p\in D[t', S\setminus \{v\}]$ a copy which includes $v$.
This can be achieved in $\bigO(\pmax)$, since we only maintain the costs of every solution and only need to add a vector of constant size, depending on $S$ and $V$.
\end{proof}

\ForgetNodesLemmaST*

\begin{proof}
Let $p\in \mathcal{P}_t^S$ be an arbitrary Pareto-optimal solution. 
If $v\notin p$, then $p$ must be contained in $D[t', S]=\mathcal{P}_{t'}^S$.
If $v\in p$, then $p$ must be contained in $D[t', S \cup \{v\}] = \mathcal{P}_{t'}^{S\cup \{v\}}$.
Therefore, the union of the sets $D[t', S\cup \{v\}]$ and $D[t', S \setminus \{v\}]$ contains all candidates for $D[t,S]$.
Since $p\in \PP_t^S$, it cannot be dominated by any other solution in this union and will remain in the final set after removing dominated solutions.

To compute $D[t,S]$, we compute the set $\PP_{t'}^{S\cup \br{v}} + \PP_{t'}^{S\setminus \br{v}}$. 
Since both sets have size at most $\pmax$, and we only store the cost vectors of each solution, the final set can be computed in time $\bigO(\pmax \log^{\dd-2}(\pmax))$, where $\dd$ is the number of objectives.
\end{proof}

\LocalExchangeCuts*

\begin{proof}
Assume we are given two solutions $p_1, p_2 \subseteq V$ with respective costs $\co'(\delta(p_1))$ and $\co'(\delta(p_2))$, and they are identical in their intersection with points in bag $X_t$, i.e., $p_1\cap X_t=p_2\cap X_t$.
Let $S:= X_t\cap p_1 = X_t \cap p_2$.
If $p_1\cap V_t \prec p_2 \cap V_t$ in the instance $(t, S)$, we can improve the solution $p_2$ by replacing its subsolution in $V_t$ with that of $p_1$.
Let $\tilde{p}:= (p_1\cap V_t) \cup (p_2\setminus V_t)$ be the new constructed solution.
For the cost of $\tilde{p}$ we get
\begin{align*}
\co'(\delta(\tilde{p})) &= \co(I_{p_1}, \overline{I_{p_1}}) + \co(S, \overline{S}) + \co(O_{p_2}, \overline{O_{p_2}}) \\
& + \co(S, \overline{I_{p_1}}) + \co(I_{p_1}, \overline{S}) + \co(S, \overline{O_{p_2}}) + \co(O_{p_2}, \overline{S}) \\
& + \co(s, \overline{O_{p_2}}) + \co(s, \overline{S}) + \co(s, \overline{I_{p_1}}) \\
& + \co(O_{p_2}, t) + \co(S, t) + \co(I_{p_1}, t).
\end{align*}
Since $p_1\cap V_t \prec p_2 \cap V_t$ we have
\begin{align*}
&\co'(\delta(p_1\cap V_t))\leq \co'(\delta(p_2 \cap V_t)) \\
\Leftrightarrow~ &\co(I_{p_1} \cup S, \overline{I_{p_1}} \cup \overline{S}) + \co(S, V\setminus V_t) + \co(s,\overline{I_{p_1}}\cup \overline{S} \cup (V\setminus V_t)) + \co(I_{p_1}\cup S,t) \\
\leq~& \co(I_{p_2} \cup S, \overline{I_{p_2}} \cup \overline{S}) + \co(S, V\setminus V_t) + \co(s, \overline{I_{p_2}} \cup \overline{S} \cup (V\setminus V_t)) + \co(I_{p_2}\cup S, t) \\
\Leftrightarrow~ & \co(I_{p_1} \cup S, \overline{I_{p_1}} \cup \overline{S}) + \co(S \cup O_{p_2}, \overline{O_{p_2}}) + \co(O_{p_2}, \overline{S})+\co(s, \overline{I_{p_1}} \cup \overline{S} \cup \overline{O_{p_2}})\\
+&~ \co(I_{p_1} \cup S \cup O_{p_2}, t) \\
\leq~& \co(I_{p_2} \cup S, \overline{I_{p_2}} \cup \overline{S}) + \co(S \cup O_{p_2}, \overline{O_{p_2}}) + \co(O_{p_2}, \overline{S}) + \co(s, \overline{I_{p_2}} \cup \overline{S} \cup \overline{O_{p_2}})\\
+&~\co(I_{p_2} \cup S \cup O_{p_2}, t) \\
\Leftrightarrow~ & \co'(\delta(\tilde{p})) \leq \co'(\delta(p_2)).
\end{align*}
Since $p_1\cap V_t \prec p_2 \cap V_t$, there exists at least one objective $j\in [\dd]$ where we have a strict inequality.
Therefore, if a solution $p\subseteq V$ is Pareto-optimal, each subset $p\cap V_t$ has to be Pareto-optimal in the instance $(t, p\cap X_t)$ as well.
\end{proof}

\JoinNodesLemmaST*

\begin{proof}
We define for a solution $p\subseteq V_t$ with $p \cap X_t = S$ and $i\in [2]$ the following sets: $I_p^i:= p\cap (V_{t_i} \setminus X_t)$, $\overline{I_{p}^i}:=(V_{t_i}\setminus X_t)\setminus p$, $\overline{S}:= X_t\setminus S$ and $O_p:=V\setminus (V_{t_1}\cup V_{t_2})$.
We can write the cost $\co'(\delta(p))$ in terms of its subsolutions $I_p^1, I_p^2$, $S$ and $O_p$, minus some constant depending only on $S$:

For $\co'(\delta(p))$ we get
\begin{align*}
\co'(\delta(p)) &= \co(I_p^1 \cup I_p^2 \cup S \cup \{s\}, \overline{I_p^1} \cup \overline{I_p^2} \cup \overline{S} \cup O_p \cup \{t\}) \\
&= \co(I_p^1, \overline{I_p^1} \cup \overline{S}) + \co(S, \overline{I_p^1}) + \co(S, \overline{I_p^2}) \\
&+ \co(S, \overline{S} \cup O_p) + \co(s, \overline{I_p^1} \cup \overline{I_p^2} \cup \overline{S} \cup O_p) + \co(I_p^1 \cup I_p^2 \cup S, t).
\end{align*}

For the sum of its partial solutions we get
{\small{
\begin{align*}
&\co'(\delta(p \cap V_{t_1})) + \co'(\delta(p \cap V_{t_2}))\\
=~& \co((p \cap V_{t_1}) \cup \{s\}, V\setminus (p \cap V_{t_1}) \cup \{t\}) + \co((p \cap V_{t_2}) \cup \{s\}, V\setminus (p \cap V_{t_2}) \cup \{t\}) \\
=~& \co(I_p^1 \cup S \cup O_p \cup \{s\}, \overline{I_p^1} \cup (V_{t_2}\setminus X_t) \cup \overline{S} \cup O_p \cup \{t\}) \\
+~& \co(I_p^2 \cup S \cup O_p \cup \{s\}, \overline{I_p^2} \cup (V_{t_1}\setminus X_t) \cup \overline{S} \cup O_p \cup \{t\}) \\
=~& \co(I_p^1 \cup S \cup \{s\}, \overline{I_p^1} \cup \overline{S} \cup \{t\}) + \co(S \cup \{s\}, (V_{t_2} \setminus X_t) \cup \{t\}) + \co(S \cup \{s\}, O_p \cup \{t\}) \\
+~& \co(I_p^2 \cup S \cup \{s\}, \overline{I_p^2} \cup \overline{S} \cup \{t\}) + \co(S \cup \{s\}, (V_{t_1} \setminus X_t) \cup \{t\}) + \co(S \cup \{s\}, O_p \cup \{t\})
\end{align*}
}}

\noindent If we compute $\co'(\delta(p \cap V_{t_1})) + \co'(\delta(p \cap V_{t_2})) - \co'(\delta(p))$ we get exactly the cost $\co'(\delta(S))$.
Therefore, for a fixed selection $S\subseteq X_t$, we can express the cost of any solution $p\subseteq V_t$ as the sum of the costs of its parts in $V_{t_1}$ and $V_{t_2}$, minus the cost the cut induced by $S$ alone.
By Lemma~\ref{lemma:LocalExchangeCuts}, we know that the solutions $p\cap V_{t_1}$ and $p\cap V_{t_2}$ must also be Pareto-optimal in the instances $(t_1, S)$ and $(t_2,S)$, respectively.
By our induction hypothesis, these parts must be contained in $D[t_1, S]$ and $D[t_2, S]$. 
Since our dynamic program considers all combinations of solutions from these two sets and removes only those that are dominated, the solution $p$ must also be computed in the final set $D[t,S]$.

For a given selection $S\subseteq X_t$, we therefore compute the combined set $\PP_{t_1}^S \oplus \PP_{t_2}^S$. 
Since both sets are are of size at most $\pmax$, the number of pairs is bounded by $\pmax^2$. 
If we only consider the respective costs and combine every solution pair by adding their costs and subtract $\co'(\delta(S))$, as described in Section~\ref{sec:Preliminaries}, we can compute this set in 
$\bigO(\pmax^2 \log^{\max \{\dd-2,1\}}(\pmax^2))$.

\end{proof}
Finally, we discuss how to retrieve a concrete solution if, for each instance, we only store the cost vectors. 
As previously mentioned, it suffices to store, for each non-leaf node $t$ and each cost vector $v\in \PP_t^{\mathcal{S}}$, a pointer (or two in case of a join node) to a solution $p'\in \PP_{t'}^{\mathcal{S}'}$, where $t'$ is a child of $t$ and $p'$ was used in the construction of the solution associated with $p$. 
Setting these pointers can be done in the natural way, as already indicated in the respective proofs for each node type.

Since a solution $p$ with cost vector $\co(p)=v$ corresponds to a subset of vertices, and vertices can only be added at introduce nodes, we can reconstruct the full solution $p$ by following the pointers back through the tree and consider each relevant introduce node.
If a vertex $u$ is newly introduced at a node and appears in the selection $\mathcal{S}$, we can conclude that $u\in p$, and include it in the reconstructed solution.

\section{Exponential size Pareto sets for the bicriteria aggregation problem
}\label{apx:proofKP}

In this section, we present a family of instances for the bicriteria triangle aggregation problem that has exponentially many Pareto-optimal solutions. 

In the knapsack problem, we are given $n$ items with profits $p=(p_1,\ldots,p_n)^T\in \RR^n_{>0}$ and weights $w=(w_1,\ldots,w_n)^T\in \RR^n_{>0}$, along with a capacity $T$. 
In the single-criterion version, the goal is to find a subset $I\subseteq [n]$ such that $p(I)=\sum_{i\in I} p_i$ is maximized and $w(I)=\sum_{i\in I} w_i \leq T$. In the bicriteria knapsack problem, we ignore the capacity and have two (conflicting) objectives instead: maximize $p(I)$ and minimize $w(I)$. If one chooses $p_i=w_i=2^i$, one obtains an instance where each of the $2^n$ solutions $I\subseteq\{1,\ldots,n\}$ is Pareto-optimal.

For two polygons $s_1,s_2\in T\cup P$, we define $l(s_1,s_2)$ as the length of the edge separating $s_1$ and $s_2$ if they are adjacent, and $0$ otherwise. 
Additionally, let $Z=\{z\}$ denote the remaining area not covered by triangles or polygons.
Observe that the total perimeter of a solution $S\subseteq T$ equals the sum of all edge lengths between polygons in $S\cup P$ and polygons in $(T\setminus S) \cup Z$.
Thus, we can express the perimeter of a solution $S\subseteq T$ as
\begin{align*}
P(S) = \sum_{s\in S\cup P} \sum_{s'\in (T\setminus S) \cup Z} l(s,s').
\end{align*}

\begin{figure}[t]
\centering
\includegraphics[width=.7\linewidth]{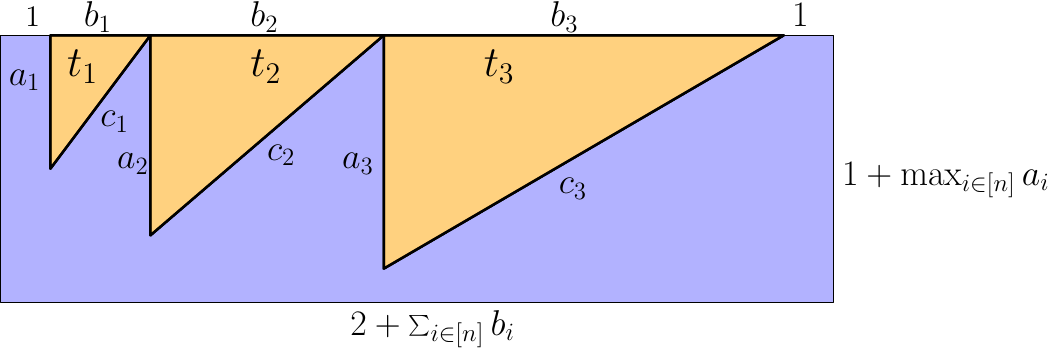}
\caption{Example transformation of a knapsack instance to a triangle instance.}
\label{kp:instance}
\end{figure}

\KPTA*

\begin{proof}
Let $(p,w)$ be an instance for the bicriteria knapsack problem. 
We construct a corresponding instance of the triangle aggregation problem as follows: 
We have a single polygon $P=\{p\}$ and a set of $n$ triangles $T=\{t_1,\ldots,t_n\}$, where each triangle $t_i$ is adjacent to polygon $p$ along two of its edges, and its third edge is connected with the outside area $z$. 
We assume each triangle $t_i\in T$ is right-angled with catheti $a_i, b_i$ and hypotenuse $c_i$, such that only the edge of length $b_i$ is adjacent to the outside area $z$. 
We construct the instance such that choosing triangle $t_i$ increases the total area by $w_i$ and reduces the perimeter by $p_i$.
Specifically, we set $a_i,b_i$ and $c_i= \sqrt{a_i^2+b_i^2}$ such that
\begin{align*}
w_i &= \frac{a_i\cdot b_i}{2} = A(t_i) \qquad\text{and}\qquad
p_i = a_i+\sqrt{a_i^2+b_i^2} - b_i = a_i+c_i-b_i.
\intertext{Solving both equalities for $a_i$ and $b_i$ yields }
a_i &=\frac{2w_i}{b_i} \qquad\text{and}\qquad
b_i = \frac{-p_i^2+4w_i + \sqrt{p_i^4+24p_i^2w_i+16w_i^2}}{4p_i}.
\end{align*}
Since $p_i>0$ and $w_i>0$ for all $i\in [n]$, the values $a_i,b_i,c_i$ are strictly positive as well. 
Now, any subset $I\subseteq [n]$ representing a solution of the knapsack instance yields a corresponding subset of triangles $T_I=\{t_i\in T \mid i \in I\}$ in the polygon aggregation instance. 
The solution $T_I$ has area $A(T_I)=\sum_{i\in I}w_i + A(P)$ and perimeter $P(T_I)=B-\sum_{i\in I}p_i$ where $B\in \mathbb{R}_{> 0}$ denotes the  perimeter of the polygon $p$. 
See Figure~\ref{kp:instance} for an example of this transformation. 
It follows that every (Pareto-optimal) solution of the knapsack problem corresponds to a (Pareto-optimal) solution of the polygon aggregation instance, and vice versa.
\end{proof}

\section{Multiobjective MST}\label{apx:proofsMST}

In the multiobjective minimum spanning tree problem, we are given a graph $G=(V,E)$ and a cost function $\co:E\rightarrow \mathbb{R}^\dd$ on the edges where $\dd\geq 2$ is a constant number of objectives. 
A spanning tree $T\subseteq E$ is associated with the cost vector $\co(T):=\sum_{e\in T}\co(e)$, and the goal is to find all Pareto-optimal spanning trees. 

Our algorithm for computing the set of Pareto-optimal solutions is largely based on the algorithm for computing an optimal Steiner tree presented in~\cite{CF15}. 
We first compute a nice tree decomposition with introduce edge nodes (for a definition of a nice tree decomposition with introduce edge nodes, see \cref{apx:introduce_edges}), and each bag induces a set of subproblems. The difference compared to single-criterion optimization problems is that we do not compute an optimal solution for each subproblem, but the set of Pareto-optimal solutions for each subproblem.

Assume we are given a nice tree decomposition $(\TT, \{X_t\}_{t\in V(\TT)})$ of the input graph with treewidth $\Tw$.
As in the case of the multiobjective $s$-$t$ cut problem, we assume that the dynamic programming table entries store the actual solution, but the procedure of reconstructing the actual solution from only the cost vectors can analogously be applied.

For each node $t\in V(\mathbb{T})$ and possible partition $\mathcal{S}= \{S_1,\ldots,S_q\}$ of the bag $X_t$ into nonempty subsets, we say that a forest $P\subseteq E_t$ on $V_t$ is \textit{compatible} with $(t,\mathcal{S})$ if the following property holds: for any pair of vertices $u,v\in X_t$, the vertices $u$ and $v$ lie in the same connected component of $P$ if and only if there exists some $i\in [q]$ such that $u\in S_i$ and $v\in S_i$. 
This means that $\mathcal{S}$ encodes which vertices of the  bag should be in the same connected component. 
Let $C_t^{\mathcal{S}}$ denote the set of all forests $P\subseteq E_t$ that are compatible with $(t,\mathcal{S})$.

For each node $t\in V(\mathbb{T})$ and possible partition $\mathcal{S}$ 
let $\mathcal{P}_t^\mathcal{S}:=\{P \in C_t^\mathcal{S} \mid \nexists Q \in C_t^\mathcal{S}: Q \prec P \}$ be the set of Pareto-optimal solutions that are compatible with $(t,\mathcal{S})$. 
Let $\pmax := \max_{t\in V(\mathbb{T}),\mathcal{S}} \vert \mathcal{P}_t^\mathcal{S} \vert$ denote the maximum size of any Pareto set over all nodes and partitions.

We apply our dynamic programming approach to the nice tree decomposition using dedicated subroutines for the different node types to compute the final set of Pareto-optimal solutions.
The algorithm maintains a table with entries $D[t,\mathcal{S}]$, where for each node $t\in V(\mathbb{T})$ and each possible partitions $\mathcal{S}$, the entry stores the cost $\co(P_i)$ of all solutions $P_i\in \PP_t^{\mathcal{S}}$. 
Since we assume the number $\dd$ of objectives as a constant, each cost vector $\co(P_i)$ has constant size. 
Thus, each table entry $D[t,\SSS]$ requires $\bigO(\pmax)$ space.

The following lemma is important for the correctness of the dynamic program and the subroutines that we describe below.

\begin{restatable}{lem}{LocalExchangeMST}\label{lemma:LocalExchange}
Let $t\in V(\mathbb{T})$ be an arbitrary node in the tree decomposition. Furthermore, let $P_1,P_2\subseteq E$ be two solutions such that $P_1\cap E_t\in C_t^\mathcal{S}$ and $P_2\cap E_t\in C_t^\mathcal{S}$ for some arbitrary partition $\mathcal{S}$. If $P_1\cap E_t \prec P_2\cap E_t$ then we can improve $P_2$ by replacing $P_2\cap E_t$ with $P_1\cap E_t$.
\end{restatable}

\begin{proof}
Let $\tilde{P}:=(P_2 \setminus E_t)\cup (P_1\cap E_t)$ denote the solution obtained from $P_2$ by replacing the edges in $P_2\cap E_t$ with those in $P_1\cap E_t$. 
We have $\tilde{P}\prec P_2$, since
\[
  \co(\tilde{P}) = \co(P_2 \setminus E_t)+\co(P_1\cap E_t)
               \le \co(P_2 \setminus E_t)+\co(P_2\cap E_t)
               = \co(P_2)
\] 
where the inequality represents a componentwise inequality over all cost vector entries and is strict at at least one component.

It remains to show that~$\tilde{P}$ is a spanning tree. 
Since $P_1\cap E_t$ and $P_2\cap E_t$ induce the same number of connected components on $V_t$, they must contain the same number of edges. 
Hence, we only need to show that~$\tilde{P}$ does not contain a cycle. 
Suppose for contradiction that $\tilde{P}$ contains a cycle~$C$. Since both $P_1\cap E_t$ and $P_2 \setminus E_t$ are forests, any such cycle must contain edges from both sets. 
We can therefore decompose $C$ into segments $C=(C_1^1,C_1^2,C_2^1,C_2^2,\ldots,C_{\ell}^1,C_{\ell}^2)$, where each $C_i^1$ only contains edges from $P_1\cap E_t$ and each $C_i^2$ only contains edges from $P_2\setminus E_t$. 

Since all edges incident to vertices in $V_t\setminus X_t$ are contained in $E_t$, the first and the last vertex of every segment $C_i^1$ must be from~$X_t$. 
Because $P_1\cap E_t$ is contained in $C_t^\mathcal{S}$, the endpoints of each such segment are in the same class of the partition $\mathcal{S}$. 
Since $P_2\cap E_t$ is also contained in $C_t^\mathcal{S}$, there must exist a path ${C_i^1}'$ in $P_2\cap E_t$ connecting the same pair of endpoints. 
By replacing each segment $C_i^1$ with the corresponding path ${C_i^1}'$, we obtain a new cycle $C'=({C_1^1}',C_1^2,{C_2^1}',C_2^2,\ldots,{C_{\ell}^1}',C_{\ell}^2)$, which is entirely contained in $P_2$. 
This contradicts the assumption that $P_2$ is a spanning tree.
\end{proof}

Next we describe the procedure for each type of node in our nice tree decomposition and their respective proofs. 

\smallskip
\noindent
$\leafNode(t,\mathcal{S})$: For a leaf node $t$, the graph $G_t$ does not contain any vertices or edges. Therefore, there is only one possible solution: $D[t, \{\emptyset\}]=\{\emptyset\}$.

\smallskip
\noindent
$\intrNode(t,\mathcal{S})$: Consider some node $t$ and its only child $t'$ with $X_t = X_{t'} \cup \{v\}$. 
We want to compute the set $D[t,\mathcal{S}]$. 
Since the vertex $v$ is only introduced at node $t$, the graph $G_t$ does not contain any edge incident to~$v$. 
As a consequence, if the partition $\mathcal{S}$ does not contain $v$ in its own separate class, but in a class with at least one other vertex, the set $C_t^\mathcal{S}$ is empty. 
This implies $D[t,\mathcal{S}]=\{\emptyset\}$. 
On the other hand, if $v$ is in its own class in $\mathcal{S}$, we define $D[t,\mathcal{S}]=D[t',\mathcal{S}\setminus \{\{v\}\}]$.

\begin{restatable}{lem}{IntroduceNodesLemmaMST}
If node $t$ introduces some vertex $v$ and has a child $t'$ for which all possible sets $D[t',\mathcal{S}]$ have been computed, then $\intrNode(t,\mathcal{S})$ computes $D[t,\mathcal{S}]$ in time $\bigO(\pmax)$ for any possible partition $\mathcal{S}$.
\end{restatable}

\begin{proof}
Let $P\subseteq E_t$ be an arbitrary Pareto-optimal solution, and let $\mathcal{S}$ be the partition induced by $P$ in $G_t$, i.e., $P\in C_t^{\mathcal{S}}$. 
Since $v$ is introduced at node $t$ and no edge in $E_t$ is incident to $v$, the vertex $v$ must form a singleton class in the partition $\SSS$.
If $\mathcal{S}' = \mathcal{S}\setminus \{\{v\}\}$, the solution $P\cap E_{t'}$ must be contained in $C_{t'}^{\SSS'}$, and thus $P'\in D[t',\mathcal{S}']=\mathcal{P}_{t'}^\mathcal{S'}$. 
\end{proof}

\smallskip
\noindent
$\forgNode(t,\mathcal{S})$: Let $t$ be a forget node with  child $t'$ such that $X_t = X_{t'} \setminus \{v\}$, and let $\mathcal{S}=\{S_1,\ldots,S_{\ell}\}$ be a partition of $X_t$. 
To compute $D[t, \mathcal{S}]$, we compute the set of Pareto-optimal solutions from the union of all sets $D[t',\mathcal{S}_i]$ for $i\in[\ell]$, where $\mathcal{S}_i=\{S_1,\ldots,S_{i-1},S_i\cup \{v\},S_{i+1},\ldots,S_{\ell}\}$.
That is, $\mathcal{S}_i$ is the same partition as $\mathcal{S}$, but in set $S_i$ we additionally include $v$.

\begin{restatable}{lem}{ForgetNodesLemmaMST}
If node $t\in V(\mathbb{V})$ removes some vertex $v$ and has a child $t'$ for which all possible sets $D[t', \mathcal{S}]$ have been computed, then 
$\forgNode(t,\mathcal{S})$ computes $D[t,\mathcal{S}]$ in time $\bigO(\Tw \cdot p_{\max} \log^{\max\{\dd-2,1\}}(\Tw \cdot p_{\max}))$ for any possible partition $\mathcal{S}$.
\end{restatable}

\begin{proof}
Let $P\subseteq E_t$ be an arbitrary Pareto-optimal solution, and let $\mathcal{S}$ be its  associated partition in $G_t$, i.e., $P\in C_t^{\mathcal{S}}$. 
Let $F_1,\ldots,F_{\vert \mathcal{S} \vert}$ be the connected components of the forest induced by $P$. 
Since every solution has to be a spanning tree, vertex $v$ must be connected to the rest of the graph via some edge $\{u,v\}$ with $u\in V_t$, otherwise $v$ would be isolated and $P$ does not span all vertices.
Therefore, $v$ must belong to exactly one of the components $F_1,\ldots,F_{\vert \mathcal{S} \vert}$, and all edges incident to~$v$ are contained in $E_t$.

Assume $v\in F_i$. 
Then $P \in C_{t'}^{\mathcal{S}_i}$, where $\mathcal{S}_i$ is obtained by inserting $v$ into the part $S_i\in \mathcal{S}$ that corresponds to $F_i$.
Since $P$ is Pareto-optimal, it must also belongs to the corresponding set $D[t',\mathcal{S}_i]=\mathcal{P}_{t'}^{\mathcal{S}_i}$. 
Thus, $P$ is contained in the union from which we construct $D[t,\mathcal{S}]$, and will not be removed, in the filtering step when removing non-Pareto-optimal solutions.

We need to compute the Pareto set $\sum_{i=1}^{\vert X_t \vert} \PP_{t'}^{\SSS_i}$. We can upper bound $\vert X_t \vert \leq \Tw + 1$ and $\vert \PP_t^{\SSS_i} \vert \leq \pmax$ and, as already discussed in \cref{apx:pareto_runtime}, obtain a running time of $\bigO(\Tw \cdot \pmax\cdot \log^{\max\{\dd-2,1\}} (\Tw \cdot \pmax))$ for $\dd\geq 2$.

When working solely with the cost vectors of each solution, the computation of the new Pareto set follows in the same manner: 
For each possible partition $\SSS_i$, we collect all sets $D[t', \SSS_i]$ into a single set and remove the dominated solutions. 
As argued before, for every Pareto-optimal solution $P_i\subseteq E_t$ with cost vector $v_i=\co(P_i)$, there exists a corresponding solution $P_j\subseteq E_{t'}$ with $v_j=\co(P_j)=v_i$, where $P_j$ is Pareto-optimal for some partition $\SSS_{b}$ from which $P_i$ originated from. 
Since we assume that all entries $D[t',\mathcal{S}']$ have been correctly computed for all partitions $\mathcal{S}'$, each value $v\in D[t', \mathcal{S}']$ must be the cost of some Pareto-optimal solution.
As a result, the value $v_j$ associated with $P_j$ cannot be dominated by any other cost vector and will therefore remain in the final Pareto set. 
We can thus include $v_j$ into our final set of Pareto-optimal solutions and store a corresponding pointer to its originating solution.
\end{proof}

\smallskip
\noindent
$\intrEdgeNode(t,\mathcal{S})$:
Let $t$ be an introduce-edge node with child $t'$ such that $X_t = X_{t'}$, and suppose edge $e=\{u,v\}$ with $u,v\in X_t$ is introduced at $t$. 
If $u$ and $v$ belong to different classes in the partition $\mathcal{S}$, then $e$ cannot be part of any solution in $C_t^\mathcal{S}$, as it would connect components required to be disconnected.
In this case, the solution space remains unchanged and we set $D[t,\mathcal{S}] = D[t', \mathcal{S}]$. 
If $u$ and $v$ belong to the same class $S_i\in \mathcal{S}$, we consider every way of splitting $S_i$ into two disjoint nonempty sets $S_i^1,S_i^2$ such that $S_i^1\cup S_i^2 = S_i$, $S_i^1 \cap S_i^2 = \emptyset$, $u\in S_i^1$ and $v\in S_i^2$.
Each such split defines a new partition $\mathcal{S}_i'$ obtained by replacing $S_i$ with $S_i^1$ and $S_i^2$. 
For each solution $P\in D[t',\mathcal{S}_i']$, we construct a new candidate solution $P\cup \{e\}$, which merges the two components via the edge $e$. 
Additionally, we also consider all solutions in $D[t',\mathcal{S}]$ and at last filter out all dominated solutions from the union to obtain $D[t,\mathcal{S}]$.

\begin{restatable}{lem}{IntroduceEdgesLemmaMST}
If node $t\in V(\mathbb{T})$ introduces edge $e=\{u,v\}$ and has a child $t'\in V(\mathbb{T})$ for which all possible sets $D[t', \mathcal{S}]$ have been computed, then $\intrEdgeNode(t,\mathcal{S})$ computes $D[t,\mathcal{S}]$ in time $\bigO(2^{\Tw} \cdot p_{\max}\log^{\max\{\dd-2,1\}}(2^{\Tw} \cdot p_{\max}))$ for any possible partition $\mathcal{S}$.
\end{restatable}

\begin{proof}
Assume $P\subseteq E_t$ is a Pareto-optimal solution for the partition $\mathcal{S}=(S_1,\dots,S_\ell)$. 
We distinguish two cases depending on wether the introduced edge is contained in $P$ or not. 
If $e\notin P$, then $P\subseteq E_{t'}$, and since $P$ is Pareto-optimal for partition $\mathcal{S}$, we must have $P \in D[t', \mathcal{S}]$.
As we include all solutions from $D[t', \mathcal{S}]$ in the computation of $D[t,\mathcal{S}]$, and only remove dominated solutions, it follows that $P$ will be contained in $D[t,\mathcal{S}]$.
Now assume $e\in P$.
Then $u$ and $v$ must be in the same connected component of $P$, and thus in the same class $S_i$.
By removing edge $e$ from $P$, we obtain a the forest $P':=P\setminus \{e\}$, which separates $u$ and $v$ into different components.
Let $\mathcal{S}'$ be the partition of $X_t$ induced by $P'$.
This partition is identical to $\mathcal{S}$, except that the class $S_i$ is replaced by two disjoint sets $S_i^1$ and $S_i^2$, where $u\in S_i^1,v\in S_i^2$, and $S_i^1 \cup S_i^2=S_i$.
Since $P' \subseteq E_{t'}$, and by assumption $D[t', \mathcal{S}']$ contains all Pareto-optimal solutions for $\mathcal{S}'$, we must have $P' \in D[t', \mathcal{S}']$, otherwise some solution $Q \in D[t', \mathcal{S}']$ would dominate $P'$, so the solution $Q\cup \{e\}$ would dominate $P$, which is a contradiction. 
Because we consider all possible partitions $\mathcal{S}_i'$ of $\mathcal{S}$, the specific partition $\mathcal{S}'$ will be among them. 
Hence, we include $P=P'\cup \{e\}$ into the set of candidate solutions and, due to its Pareto-optimality, will not remove it in the filtering step.

To analyze the running time, we first bound the number of sets $D[t',\mathcal{S}_i']$ considered in this step. 
Since $\vert S_i \vert \le \Tw+1$, there are at most $2^{\Tw-1}$ different ways to split $S_i$ into two disjoint subsets $S_i^1$ and $S_i^2$ such that $u\in S_i^1$ and $v\in S_i^2$. 
Hence, we consider at most $2^{\Tw - 1}$ different partitions $\mathcal{S}_i'$, which, together with the original partition $\mathcal{S}$, gives a total of at most $2^{\Tw-1}+1 \leq 2^{\Tw}$ partitions.

Since each corresponding table entry contains at most $\pmax$ solutions, the union of all these sets has size at most $\pmax\cdot 2^{\Tw}$. 
For $n=\pmax \cdot 2^{\Tw}$, we know from \cref{apx:pareto_runtime} that this gives us a final runtime bound of $\bigO (2^{\Tw}\cdot \pmax \cdot \log^{\max\{\dd-2,1\}}(2^{\Tw} \cdot \pmax))$.

The procedure to check if two sets $\SSS_1,\SSS_2$ correctly split a partition $\SSS_i$ is independent of the actual cost values stored in the DP table.
Therefore, if we represent each solution by only the respective cost vector, this step of combining candidate sets and removing dominated solutions remains unaffected.
Additionally, we can store a pointer for each cost vector to the solution it originated from, allowing to reconstruct the full solution later.
\end{proof}

\smallskip
\noindent
$\joinNode(t,\mathcal{S})$:
Let $t$ be a join node with children $t_1$ and $t_2$, such that $X_t = X_{t_1} = X_{t_2}$. 
For each $i\in [2]$, let $G_t^{\SSS_i}=(X_t,E_t^{\SSS_i})$ be an arbitrary forest whose connected components correspond exactly to the partition $\SSS_i$.
We define the combined multigraph $G_t^{\SSS_1, \SSS_2}=(X_t, E_t^{\SSS_1}\cup E_t^{\SSS_2})$ obtained by combining the edge sets from both forests.

We say that two partitions $\mathcal{S}_1$ and $\mathcal{S}_2$ \textit{create} the partition $\mathcal{S}$ if 
the connected components of $G_t^{\mathcal{S}_1,\mathcal{S}_2}$ correspond exactly to $\mathcal{S}$ and contains no cycles.
To compute the Pareto-optimal solutions for partition $\mathcal{S}$, we consider every pair of partitions $\mathcal{S}_1, \mathcal{S}_2$ such that they create $\mathcal{S}$. 
For each such pair, we take the union of every solution $p_1\in D[t_1,\mathcal{S}_1]$ and $p_2\in D[t_2,\mathcal{S}_2]$ to form a new solution $p=p_1\cup p_2$, and finally, remove all dominated solutions from this union.

\begin{restatable}{lem}{JoinNodesLemmaMST}
If node $t\in V(\mathbb{T})$ is a join node and has two children $t_1,t_2\in V(\mathbb{T})$ and we correctly computed all possible sets for $D[t', \mathcal{S}]$, then $\joinNode(t,\mathcal{S})$ computes $D[t,\mathcal{S}]$ in time $\bigO((2\Tw)^{\bigO(\Tw)}\cdot \pmax^2 \log^{\max\{\dd-2,1\}}((2\Tw)^{\bigO(\Tw)}\cdot \pmax^2))$.
\end{restatable}

\begin{proof}
Let $P\subseteq E_t$ be an arbitrary Pareto-optimal solution in $\mathcal{P}_{t}^{\mathcal{S}}$. 
Let $\mathcal{S}_1$ and $\mathcal{S}_2$ denote the partitions of $X_t$ induced by $P\cap E_{t_1}$ and $P\cap E_{t_2}$ in the graphs $G_{t_1}$ and $G_{t_2}$, respectively. 
The partitions $\mathcal{S}_1$ and $\mathcal{S}_2$ must create the partition $\mathcal{S}$.

As in the previous proofs, if $P\cap E_{t_1}$ were not  Pareto-optimal for the compatible instance $(t_1,\mathcal{S}_1)$, we could replace $P\cap E_{t_1}$ with some solution $Q\in D[t_1,\mathcal{S}_1]$ that dominates $P$ (cf.\ Lemma~\ref{lemma:LocalExchange}). 
The combined solution $Q\cup (P\setminus E_{t_1})$ would belong to $C_t^{\mathcal{S}}$ and dominates $P$. 
The same argument applies for $P\cap E_{t_2}$.
By considering every feasible combination of partitions that create $\mathcal{S}$, the solution pair $P\cap E_{t_1}$ and $P\cap E_{t_2}$ must at some point be considered. 
Since this combination corresponds to the original solution $P$, which is Pareto-optimal, it will not be removed when filtering out dominated solutions from the union.

We can bound the running time by the number of possible combinations, the time to check if a configuration is feasible, and the time needed to remove dominated solutions. 
There are at most $((\Tw+1)^{\Tw+1})^2 \leq (2\Tw)^{\bigO(\Tw)}$ possible combinations for $\mathcal{S}_1$ and $\mathcal{S}_2$, and checking a pair takes time $\Tw^{\bigO(1)}$.

Since the procedure to check if two partitions $\SSS_1,\SSS_2$ create the partition $\SSS$ is independent of the actual solutions, we can safely replace the process of combining all actual respective Pareto-optimal solutions with combining only the respective cost vectors and then remove the dominated ones. 
For each solution $v$, formed by the cost vector $v_1\in D[t_1,\SSS_1]$ and cost vector $v_2\in D[t_2,\SSS_2]$, we maintain two pointers, one pointing to the solution cost $v_1$ and the other pointing to $v_2$.

Since $D[t_1,\mathcal{S}_1]$ and $D[t_2,\mathcal{S}_2]$ each contain at most $\pmax$ solutions, the total number of possible solution pairs is at most $\pmax^2$. 
Thus, the total number of solutions we need to consider is bounded by $(2\Tw)^{\bigO(\Tw)}\cdot \pmax^2$, and the runtime to compute these solutions is in the same order. 
From this set, the Pareto-optimal solutions can be obtained in time $\bigO((2\Tw)^{\bigO(\Tw)}\cdot \pmax^2 \log^{\max\{\dd-2,1\}}((2\Tw)^{\bigO(\Tw)}\cdot \pmax^2))$.
\end{proof}

Together, the previous lemmas imply the main result.
\begin{proof}[Proof of Theorem \ref{thm:mmst}]
The correctness of our algorithm follows immediately from the proofs of the lemmas for the different node types and from the fact that $D[r,\{\emptyset\}]$ contains the set of Pareto-optimal solutions in the graph $G_r=G$, which consists of a single connected component.

We observe that the runtime is dominated by the computation of the join nodes. 
For each join node, there are at most $(2\Tw)^{\bigO(\Tw)}$ partitions $\mathcal{S}$, and for each partition, the running time is $\bigO((2\Tw)^{\bigO(\Tw)}\cdot \pmax^2 \log^{\max\{\dd-2,1\}}((2\Tw)^{\bigO(\Tw)}\cdot \pmax^2))$. 
Since a nice tree decomposition contains $\bigO(n\Tw)$ nodes, we can bound the overall running time by the number of nodes times the computation time for all configurations at a join node. 
This results in a total runtime of $\bigO(n(2\Tw)^{\bigO(\Tw)}\cdot \pmax^2 \log^{\max\{\dd-2,1\}}((2\Tw)^{\bigO(\Tw)}\cdot \pmax^2))$.
\end{proof}

\section{Multiobjective TSP}\label{apx:proofsTSP}
We now consider the multiobjective traveling salesman problem. We are given a graph $G=(V,E)$ along  with a cost function $\co:E\rightarrow \mathbb{R}^\dd$ and aim to find all Pareto-optimal Hamiltonian cycles.
The traveling salesman problem is known to be $\mathcal{NP}$-hard, but can be solved in FPT-time with respect to its treewidth.

In the following, we briefly describe a folklore dynamic programming approach for solving the single-criterion TSP and how this algorithm can be adapted to the multiobjective case
\footnote{See Marx, D.: Lecture Notes on Fixed Parameter Algorithms (page 17), available online~\cite{marx-notes}.}.
For simplicity, as for the multiobjective $s$-$t$ cut problem, we assume that the dynamic programming table stores the full solutions, i.e., all edges contained in the current solution instead of only the overall cost. 
Since our algorithm only depends on the costs of the solutions in the subproblems, in the actual implementation this can be exchanged with just saving the cost vectors as well as pointers to the child solution (or two solutions in case of join nodes) from which the current one originated from.
The full solution can then be reconstructed in a postprocessing step by recursively following these pointers, as already discussed in more detail for the multiobjective minimum spanning tree problem in \cref{apx:proofsMST}.

Given a nice tree decomposition with introduce edge nodes (for a definition of a tree decomposition with introduce edge nodes, see \cref{apx:introduce_edges}) for an edge-weighted graph $(G,\co)$ and an arbitrary node $t$, consider the partial solution derived from intersecting any valid tour $S$ with the graph $G_t$. 
Every vertex $v\in V_t\setminus X_t$ has degree exactly 2, and each vertex $v\in X_t$ has a degree depending on how many of its neighbors are contained in $V_t$.
To reconstruct a solution, we must maintain a matching for all vertices with degree $1$, which indicates which pair of endpoints from our current bag should be the starting and ending point of a path in the solution $S$.

For each node $t\in V(\mathbb{T})$, we define a partition of the vertices in $X_t$ into three sets $D_t^0$, $D_t^1$, and $D_t^2$, along with a matching $M\subseteq D_t^1\times D_t^1$. 
A solution $S\subseteq E_t$ is valid for this subproblem, if every vertex $v\in D_t^i$ has degree exactly $i$ in $S$ and for every matching pair $(u,v)\in M$, the solution $S$ contains a path connecting $u$ and $v$.

\begin{lemma}\label{lemma:LocalExchangeTSP}
Let $t\in V(\mathbb{T})$ be an arbitrary node in the tree decomposition. 
Furthermore, let $P_1,P_2\subseteq E$ be two solutions such that for all $D_t^0,D_t^1,D_t^2\subseteq X_t$, each vertex $v\in D_t^i$ has degree $i$ in $G[V_t]$ for both $P_1$ and $P_2$.
Additionally, both $P_1$ and $P_2$ induce the same matching on the vertices with degree one.
If $P_1 \cap E_t \prec P_2 \cap E_t$, we can improve $P_2$ by replacing $P_2\cap E_t$ with $P_1 \cap E_t$.
\end{lemma}

\begin{proof}
    Let $\hat{P}:=(P_2\setminus E_t)\cup(P_1\cap E_t)$ be the new combined solution resulting from replacing $E_t\cap P_2$ with $E_t \cap P_1$.
    Since $P_1\cap E_t \prec P_2 \cap E_t$, we have
    \begin{align*}
        \co(\hat{P}) = \co(P_2\setminus E_t) + \co(P_1 \cap E_t) \leq \co(P_2\setminus E_t) + \co(P_2 \setminus E_t) = \co(P_2),
    \end{align*}
    where the inequality is taken componentwise over all $\dd$ objectives, and for at least one objective we have strict inequality.

    We now show that $\hat{P}$ is a valid tour.
    Since the partitions $D_t^0,D_t^1,D_t^2$ are identical for both $P_1$ and $P_2$, every vertex in $X_t$ has the same degree in $P_1\cap E_t$ and $P_2\cap E_t$. 
    Additionally, because both $P_1$ and $P_2$ are feasible solutions, each vertex in $V_t\setminus X_t$ has degree exactly two in both solutions.
    This implies that the number of edges in $P_1 \cap E_t$ and $P_2 \cap E_t$ are identical.
    Since $P_1$ and $P_2$ are tours, we have $\vert P_1 \vert = \vert V \vert$ and $\vert P_2 \vert = \vert V \vert$. 
    Therefore, $\hat{P}$ must also contain exactly $\vert V \vert$ edges, since they contain the same amount of edges from $E_t$ and $E \setminus E_t$.
    At last we show that the new solution $\hat{P}$ must be a cycle and visit every vertex in $V$.
    
    It remains to show that $\hat{P}$ forms a cycle, which completes the proof.
    We decompose the tour $P_2=C$ into alternating segments $C=(C_1^1,C_1^2,C_2^1,C_2^2,\dots, C_\ell^1, C_\ell^2)$, where each $C_i^1$ only contains edges from $E_t$ and each $C_i^2$ only contains edges from $E\setminus E_t$. 
    Because the vertices in $D_t^0 \cup D_t^2$ have degrees $0$ or $2$, the first and last vertex in any $C_i^j$ must be some vertex from $D_1$.
    By assumption, $P_1$ and $P_2$ induce the same matching on $D_t^1$, implying that for each segment $C_i^1\subseteq P_2$ there exists a corresponding  segment $\hat{C_i^1}\subseteq P_1$ such that the endpoints are identical. 
    By replacing any $C_i^1$ with $\hat{C_i^1}$, we obtain the new solution $(\hat{C_1^1},C_1^2,\hat{C_2^1}, C_2^2,\dots, \hat{C_\ell^1}, C_\ell^2)=\hat{P}$.
    This new tour $\hat{P}$ has the same vertex degrees in $X_t$, the same matching on $D_t^1$, and also has to correspond to a cycle. 
    Combined with our earlier observation that $\hat{P}$ contains $\vert V \vert$ edges, this implies that $\hat{P}$ is also a tour that visits every vertex.  
\end{proof}

\smallskip
\noindent
$\leafNode(t,D_t^0,D_t^1,D_t^2,M)$: For a leaf node $t$, the graph $G_t$ does not contain any vertices or edges. Therefore, there is only the empty solution: $D[t_,D_t^0,D_t^1,D_t^2,M]=\{ \emptyset \}$.

\smallskip
\noindent
$\intrNode(t,\mathcal{S})$: Consider an introduce node $t$ with child $t'$ such that $X_t = X_{t'} \cup \{v\}$. 
We want to compute the set $D[t,D_t^0,D_t^1,D_t^2,M]$. 
Since the vertex $v$ is only introduced at node $t$, the graph $G_t$ does not contain any edge incident to~$v$. 
As a consequence, a feasible solution can only exist if $v\in D_t^0$. 
If this is the case, we use the set $D[t',D_t^0\setminus \{v\}, D_t^1, D_t^2, M]$, since a solution in our current instance is Pareto-optimal if and only if the corresponding solution without $v$ is Pareto-optimal in the child instance $(t',D_t^0\setminus \{v\}, D_t^1, D_t^2, M)$. 

\begin{lemma}
    Let node $t$ introduce a vertex $v$, and let $t'$ be its child. 
    Assume that for all valid combinations the sets $D[t', {D_t^0}',{D_t^1}',{D_t^2}',M']$ have been computed, then \\$\intrNode(t,D_t^0,D_t^1,D_t^2,M)$ computes $D[t,D_t^0,D_t^1,D_t^2,M]$ in time $\bigO(\pmax)$.
\end{lemma}

\begin{proof}
    Since $v$ is newly introduced at node $t$, there does not exist any edge of the form $\{v,u\}$ for $u\in V_t$. 
    Therefore, if some solution $P$ is Pareto-optimal for the instance $(t, D_t^0,D_t^1,D_t^2, M)$, we must have $v\in D_t^0$. 
    This means that $P$ must also be Pareto-optimal for the instance $(t',D_t^0\setminus \{v\},D_t^1,D_t^2,M)$, because the sets of feasible solutions does not change when removing the isolated vertex $v$. 
    Therefore, we can set $D[t, D_t^0, D_t^1, D_t^2, M] = D[t', D_t^0 \setminus \{v\}, D_t^1, D_t^2, M]$.
    Since every Pareto set is bounded by $\pmax$ in size, and we only need to copy all solutions from some specific previous entry $D[t', D_t^0 \setminus \{v\}, D_t^1, D_t^2, M]$, the runtime $\bigO(\pmax)$ follows.
    \end{proof}

\smallskip
\noindent
$\forgNode(t,D_t^0,D_t^1,D_t^2,M)$: Let $t$ be a forget node with child $t'$ such that $X_t = X_{t'} \setminus \{v\}$, and let $(t,D_t^0,D_t^1,D_t^2,M)$ be an arbitrary instance. 
Since $v$ is removed at node $t$, it must have degree two in any feasible solution for the current instance.
Therefore, we can use the set $D[t', D_t^0,D_t^1,D_t^2\cup \{v\}]$ to compute all Pareto-optimal solutions for this instance.

\begin{lemma}
    Let node $t \in V(\mathbb{T})$ remove some vertex $v$, and let $t'$ be its child. 
    Assume that for all valid combinations, the sets $D[t',{D_t^0}', {D_t^1}', {D_t^2}', M']$ have been computed, then $\forgNode(t,D_t^0,D_t^1,D_t^2,M)$ computes $D[t,D_t^0,D_t^1,D_t^2,M]$ in time $\bigO(\pmax)$.
\end{lemma}

\begin{proof}
    Let $P$ be an arbitrary Pareto-optimal solution for the given instance $(t,D_t^0,D_t^1,D_t^2,M)$.
    Since node $t$ forgets vertex $v$, this vertex must have degree exactly two in the solution $P$, i.e., its neighbors in $P$ are vertices contained in $V_t$. 
    This implies that the solution $P$ must also be 
    Pareto-optimal for the instance $(t', D_t^0,D_t^1,D_t^2 \cup \{v\}, M)$ where for every feasible solution $v$ is required to have degree exactly two.
    Therefore, $P$ must also be contained in $D[t',D_{t}^0,D_t^1,D_t^2\cup \{v\}, M]$.  
    Since every Pareto set is bounded in size by $\pmax$, and we only need to copy all solutions from some specific previous entry $D[t',D_{t}^0,D_t^1,D_t^2\cup \{v\}, M]$, the runtime $\bigO(\pmax)$ follows.
\end{proof}

\smallskip
\noindent
$\intrEdgeNode(t,D_t^0,D_t^1,D_t^2,M)$:
Let $t$ be an introduce-edge node with child $t'$ such that $X_t = X_{t'}$, and suppose that the edge $e=\{u,v\}$ with $u,v\in X_t$ is introduced at $t$. 
If a Pareto-optimal solution does not contain $e$, then it is also feasible and Pareto-optimal for the previous instance $(t', D_t^0,D_t^1,D_t^2,M)$, and therefore contained in $D[T', D_t^0,D_t^1,D_t^2,M]$. 
If $e$ is included in a Pareto-optimal solution, removing it yields a feasible solution for $G_{t'}$, where the degrees of $u$ and $v$ are decreased by one, and the induced matching differs accordingly. 
Therefore, we consider the union over all possible ways in which, after removing $e$, the vertices $u$ or $v$ could have been matched to other degree-one vertices in $X_{t'}$ or are not matched at all. 
For each of these configurations we add $e$ to all of the corresponding solutions, take their union and remove all solutions which are not Pareto-optimal.

\begin{lemma}
Let node $t\in V(\mathbb{T})$ introduce an edge $e=\{u,v\}$, and let $t'$ be its child.
Assume that for all valid combinations, the sets $D[t',{D_t^0}', {D_t^1}', {D_t^2}', M']$ have been computed, then $\intrEdgeNode(t,D_t^0,D_t^1,D_t^2,M)$ computes $D[t,D_t^0,D_t^1,D_t^2,M]$ in time $\bigO(\Tw^2\cdot \pmax \cdot \log^{\dd-2}(\Tw^2\cdot \pmax))$.
\end{lemma}

\begin{proof}
    Let $P$ be an arbitrary Pareto-optimal solution for the instance $(t,D_t^0,D_t^1,D_t^2,M)$.
    Suppose $P$ does not contain the newly introduced edge $e=\{u,v\}$, then $P$ was already feasible for the previous instance $(t',D_t^0,D_t^1,D_t^2,M)$, which did not include $e$. 
    Therefore $P$ must be contained in the set $D[t',X_t^0,X_t^1,X_t^2,M]$. 
    Now assume $P$ contains the edge $e$. 
    We define $P_{e}=P\setminus \{e\}$. 
    Removing $e$ from $P$ will decrease the degrees of both endpoints by one. 
    Since $u$ and $v$ are contained in $X_t$, they were assigned to either $D_t^1$ or $D_t^2$, because $e\in P$ implies that their degrees must be at least one, otherwise $P$ could not have been feasible for our instance. 
    If $u$ is contained in $D_t^1$, removing $e$ will result in $u$ having degree exactly zero, so 
    we set ${D_t^1}'={D_t^1}\setminus \{u\}$ and ${D_t^0}'=D_t^0 \cup \{u\}$.
    Analogously, if $u$ is contained in $D_t^2$ we set ${D_t^2}'=D_t^2 \setminus \{u\}$ and ${D_t^1}' = D_t^1 \cup \{u\}$.
    We perform the update in an identical manner for
    vertex $v$.

    Next we will discuss how the matching needs to be updated.
    Assume first that only exactly one of the endpoints of $e$ is contained in $D_t^1$ and let this vertex be $u$.
    Let $w\in D_t^1$ be the vertex matched to $u$ in $M$, i.e., $\{u,w\}\in M$. 
    After removing $e$, the subpath in $P$ that had endpoints $u$ and $w$ now terminates at vertex $v$, since the final edge $e=\{u,v\}$ is removed. 
    Therefore we update the matching pair $\{u,w\}$ with $\{v,w\}$, i.e., we set $M'=M\setminus \{\{u,w\}\} \cup \{\{v,w\}\}$.
    The new solution $P_e$ is feasible for the instance $(t', {D_t^0}', {D_t^1}', {D_t^2}', M')$, since it fulfills the degree properties and the new matching. 
    Additionally, $P_{e}$ must be contained in the set $D[t', {D_t^0}', {D_t^1}', {D_t^2}', M']$, since otherwise there would exist another solution which dominates $P_e$ and, together with $e$, would also dominate $P$, which is a contradiction.

    Now assume in the second case that both $u$ and $v$ are contained in $D_t^1$. 
    Since their degree was one, it must be that $\{u,v\}\in M$, and in solution $P_{e}$ both have degree zero.
    Therefore we can set 
    $M' = M \setminus \{\{u,v\}\}$.
    With these modifications the solution $P_{e}$ is feasible for the new instance $(t',{D_t^0}', {D_t^1}', D_t^2, M')$. 
    By identical reasoning as before, we conclude that $P_{e}$ must also be Pareto-optimal and must be contained in $D[t',{D_t^0}', {D_t^1}', D_t^2, M']$. 

    Finally, consider the case where neither $u$ nor $v$ is contained in $D_t^1$. 
    In this case both are contained in $D_t^2$. 
    If we remove edge $e$ from $P$, both $u$ and $v$ have degree one. 
    If we now trace their respective subpaths in $P_e$, starting from either $u$ or $v$, we either visit some other vertices $p_u,p_v\in X_t$ such that either $p_u\neq p_v\in D_t^1$ or $p_u=v,p_v=,u$, so $P$ forms a tour. 
    If $p_u\neq p_v$, both vertices must be contained in the original matching, i.e., $\{p_u,p_v\}\in M$.
    If $p_u = p_v$, then the matching must have been empty, i.e., $M=\emptyset$.
    In both cases, we set ${D_t^1}'=D_t^1 \cup \{u,v\}$, ${D_t^2}' = D_t^2\setminus \{u,v\}$ and $M'=(M\setminus \{\{p_u,p_v\}\} ) \cup \{ \{u, p_u\}, \{v,p_v\} \} $.
    Then the solution $P_{e}$ is feasible for the instance $(t', D_t^0, {D_t^1}', {D_t^2}', M')$ and must be contained in the corresponding Pareto set $D[t',D_t^0, {D_t^1}', {D_t^2}', M']$.

    To construct the full set of Pareto-optimal solutions for the instance $(t,D_t^0,D_t^1, D_t^2, M)$, we take the union over all possibilities described above, add edge $e$ to every solutions in the variants where the final Pareto set should contain $e$, and then remove the dominated solutions. 
    In any case, we include the unmodified set $D[t',D_t^0,D_t^1,D_t^2,M]$ of size at most $\pmax$.
    
    We can upper bound the runtime by the last describe case, where we must try every possible pair of endpoints $p_u\neq p_v\in D_t^1$.
    Since $X_t$ has size at most $\Tw+1$, and $D_t^1$ does not contain $u$ or $v$, the set $D_t^1$ contains at most $\Tw-1$ vertices. 
    Therefore there are at most $(\Tw-1)^2+1 \leq \Tw^2$ many pairings to consider.
    Each pairing contributes at most $\pmax$ solutions, so the total size of the union is upper bounded by $\Tw^2 \cdot \pmax$.
    Computing the set of Pareto-optimal solutions from this set can be done in $\bigO(\Tw^2\cdot \pmax \cdot \log^{\dd-2}(\Tw^2\cdot \pmax))$.
\end{proof}

\smallskip
\noindent
$\joinNode(t,D_t^0,D_t^1,D_t^2,M)$:
Let $t$ be a join node with children $t_1$ and $t_2$, such that $X_t = X_{t_1} = X_{t_2}$. 
To compute the set of Pareto-optimal solutions at node $t$, we consider every possible way in which the degrees and induced matching of some Pareto-optimal solution might change for the respective subgraph $G_{t_1}$ and $G_{t_2}$.
We then take the union over all such feasible combinations and remove solutions that are not Pareto-optimal.

\begin{lemma}
    If node $t \in V(\mathbb{T})$ is a join node with two children $t_1,t_2 \in V(\mathbb{T})$ and we correctly computed all possible sets for all child instances with $t_1$ and $t_2$, then \\
    $\joinNode(t,D_t^0,D_t^1,D_t^2,M)$ computes $D[t,D_t^0,D_t^1,D_t^2,M]$ in time $\bigO(\pmax^2 \cdot (3\Tw + 3)^{\bigO(\Tw)} \cdot \log^{\dd-2}(\pmax^2 \cdot (3\Tw + 3)^{\bigO(\Tw)}))$.
\end{lemma}

\begin{proof}
    Let $P$ be some arbitrary Pareto-optimal solution for the instance $(t,D_t^0, D_t^1, D_t^2, M)$. 
    Consider the restriction of $P$ to the set of edges in $G_{t_1}$ and $G_{t_2}$, i.e., $P_{t_1}:=P\cap E_{t_1}$ and $P_{t_2}:=P\cap E_{t_2}$.

    By Lemma~\ref{lemma:LocalExchangeTSP}, we know that $P_1$ and $P_2$ must always be Pareto-optimal for their respectively induced instances.
    Otherwise, we could replace $P_1$ or $P_2$ with a solution dominating it, and obtain a new solution which dominates $P$.

    Checking if the combination of two Pareto-optimal solutions is feasible for $(t,D_t^0,D_t^1,D_t^2,M)$ only depends on the respective instance specification, not the solutions themselves.
    For two instances $(t_1, {D_{t_1}^0}', {D_{t_1}^1}', {D_{t_1}^2}', M')$ and $(t_2, {D_{t_2}^0}'', {D_{t_2}^1}'', {D_{t_2}^2}'', M'')$, 
    we check if for any vertex $v\in X_t$, its degree 
    matches the combination of degrees from the two child instances. 
    Because of our construction for a nice tree decomposition, there does not exist any edge $\{u,v\}$ with $u,v\in X_t$ at a join node, so we only need to sum up the degrees contributed by each child.
    Additionally, we check if the matchings $M'$ and $M''$ correspond to the original matching $M$.
    This can be done by creating a construction graph with vertex set $X_t$ and edge set given by $M'$ and $M''$.
    Then we check if for any matching $\{u,v\}\in M$, there exists a subpath in the construction graph with exactly $u$ and $v$ as endpoints.
    If all of these conditions are fulfilled, the Pareto-optimal solutions for both of the respective instances can be combined and yield a feasible solution for $(t,D_t^0,D_t^1,D_t^2,M)$.

    The number of possible instances for $t_1$ and $t_2$ is upper bounded by $3^{\Tw+1}\cdot (\Tw+1)^{\Tw+1} = (3\Tw+3)^{\Tw+1}$, so the number of possible combinations of pairs is upper bounded by $((3\Tw+3)^{\Tw+1})^2=(3\Tw+3)^{\bigO(\Tw)}$.
    Checking feasibility of a pair can be done in $\bigO(\Tw^2)$ using a simple BFS approach on the construction graph.
Computing $\PP_1 \oplus \PP_2$, as discussed in \cref{apx:pareto_runtime}, can be done in $\bigO (\pmax^2\cdot \log^{\max\{\dd-2,1\}}(\pmax^2))$, and the new resulting Pareto set is upper bounded in size by $\pmax^2$.
Since there are at most $\bigO((3\Tw + 3)^{\bigO(\Tw)})$ possible pairs, we have a runtime $\bigO((3\Tw + 3)^{\bigO(\Tw)})\cdot \pmax^2\cdot \log^{\max\{\dd-2,1\}}(\pmax^2)$ for computing the Pareto set of every feasible combination.
From these new Pareto sets $\PP_1,\dots,\PP_m$ with $\vert \PP_i\vert \leq \pmax^2$, we can upper bound $m=(3\Tw + 3)^{\bigO(\Tw)}$.
Therefore, we can compute the final set $\sum_{i=1}^m \PP_i$ in $\bigO(\pmax^2 \cdot (3\Tw + 3)^{\bigO(\Tw)} \cdot \log^{\dd-2}(\pmax^2 \cdot (3\Tw + 3)^{\bigO(\Tw)}))$, which is also our overall runtime. 
\end{proof}

\MTSP*

\begin{proof}
The correctness follows immediately by the correctness proofs of our previous lemmata and that our final Pareto set must be contained in $D[r, \{\}, \{\}, \{\}, \{\}]$.

The running time is dominated by the computation of a join node. Computing any entry in a join node takes $\bigO(\pmax^2 \cdot (3\Tw + 3)^{\bigO(\Tw)} \cdot \log^{\dd-2}(\pmax^2 \cdot (3\Tw + 3)^{\bigO(\Tw)}))$, and since there are $(3\Tw + 3)^{\bigO(\Tw)}$ possible entries, we can compute all entries in the same time complexity. Since the number of join nodes in a nice tree decomposition with introduce edge nodes is upper bounded by $\bigO(n \cdot \Tw)$, we get an overall running time of $\bigO (n(3\Tw+3)^{\bigO(\Tw)} \cdot \pmax^2 \cdot \log^{\max\{\dd-2,1\}}((3\Tw+3)^{\bigO(\Tw)} \cdot \pmax^2))$.

At last, like for the multiobjective $s$-$t$ cut algorithm and the multiobjective minimum spanning tree algorithm, it is enough to only save the costs of every Pareto optimal solution to reconstruct its corresponding set of edges by saving for every solution a pointer to the solution it was originating from.
\end{proof}

\section{Extended implementation details}
\label{sec:appendix_extended_implementation}

In this section we provide additional implementation details that complement the experimental section of this paper, offering a deeper insight into specific aspects of our approach.

\subsection{Graph simplification}
\label{appendix:graph_simplification}

In the $s$-$t$ cut construction for the bicriteria aggregation problem in~\cite{RDHR21}, certain vertices of the input graph $G$ represent polygons and are enforced to belong to the source side $S$ in any valid minimum $s$-$t$ cut. 
This is achieved by assigning an infinite capacity to the edge $\{s, v_i\}$ for each node $v_i$ that corresponds to a polygon. 
As a result, any feasible cut must include $v_i$ in $S$.

Since we wish to have a running time that only depends on the treewidth of the adjacency graph of the triangles,
we reduce the size of the $s$-$t$ cut graph such that only the triangles are represented as vertices in this instance.
To this end, we remove all vertices $v_i$ corresponding to the set of given polygons from the graph entirely and reflect their fixed association with $S$ through the edge weights. 
For each removed vertex $v_i$ and for every neighbor $u$ of $v_i$, we add the weight $\co(\{v_i,u\})$ to the edge $\{s,u\}$. 
If the edge $\{s,u\}$ does not yet exist, it is created with initial weight zero before the addition.
This transformation preserves the semantics of the original cut formulation, as $v_i$ is now implicitly contained in $S$, and its influence on neighboring nodes is redirected to the source $s$. 

\subsection{Saving solutions}
Let $G = (V \cup \{s,t\}, E)$ be the graph under consideration for the $s$-$t$ cut problem. 
For a given solution $p$, we define $V(p) \subseteq V$ as the subset of vertices that constitute $p$. 

For each node $t \in V(\mathbb{T})$ and selection $S \subseteq X_t$, we maintain the corresponding solution set $D[t,S]$ as a list sorted in lexicographical order. Specifically, if $D[t,S] = (p_1, \dots, p_l)$, the sorting follows the criteria $A(p_i) < A(p_j)$ and $P(p_i) > P(p_j)$ for all $i < j$. 

We store a solution $p$ in a pointer-like fashion. When a new solution $p$ is created in an introduce node by augmenting a solution $p'$ with some vertex $v$, we store the new weight of $p$, the vertex $v$, and a pointer to $p'$. 
For a new solution $p$ created in in join nodes, we store the corresponding weight along with the two solutions $p_a,p_b$ that were combined to create $p$.

This approach allows for efficient creation of new solutions while minimizing storage usage. 
In our experiments, many solutions recursively reference the same base solutions, resulting in significantly lower memory consumption compared to an approach where the vertices $V(p)$ are explicitly stored for all solutions.

\subsection{Outsourcing}
Even for small datasets the memory requirements of some subproblems can be too large to fit into RAM. 
To avoid this problem we outsource all solutions to a hard drive and only keep the solution lists currently needed for computations in RAM.
Additionally, we only load the actual vertices $V(p)$ of solutions into RAM when reconstructing the solutions at the end, as they are not needed for earlier computations. 

The outsourcing is split into two parts:
\begin{itemize}
    \item Origin-pointer file: Each time a new solution is created, we add a constant-size entry to the origin-pointer file. This entry contains all information required to recursively reconstruct the corresponding vertex set of any solution. 
    \item Surface-pointer file: Each solution list $D[t,S]$ is represented by a single file. In this file, we save the weight of each solution and the position of its entry in the origin-pointer file.
\end{itemize}

Each newly created solution \( p \) is assigned a unique identifier, starting from 0. 
This identifier corresponds to the position of the solution within the origin-pointer file. 
Each entry in the origin-pointer file has a fixed size of exactly 16 bytes. 
This constant size ensures low storage usage and enables fast retrieval of a solution by ID, since its position in the file can be directly computed via an offset of \(\text{ID} \times 16\) bytes.

The ability to maintain a compact 16-byte representation arises from the nature of different solution types. 
When a solution is generated in an introduce node, it is derived by augmenting an existing solution \( p' \) with a new vertex \( v \). 
In this case, the origin-pointer file stores the ID of \( p' \) (8 bytes) and the vertex identifier (8 bytes).
Similarly, if a solution is formed in a join node, it results from combining two solutions \( p_1 \) and \( p_2 \). 
The corresponding origin-pointer entry then consists of the respective IDs of \( p_1 \) and \( p_2 \), each requiring 8 bytes, summing to 16 bytes in total.

For new solutions in an introduce-join-forget node, the standard join structure—storing the IDs of \( p_1 \) and \( p_2 \)—remains applicable. 
However, in cases where the $\intrNode$ algorithm was skipped for specific vertices, the corresponding origin-pointer entries for \( p_1 \) and \( p_2 \) may not exist. 
Let \( I_i \) denote the set of vertices that were skipped during the introduction step on side \( i \in \{1,2\} \). If 
\(I_i \cap V(p_i) \neq \emptyset,\)
then no origin-pointer entry exists for \( p_i \); instead, an entry is only available for $p_i'$ with \(V(p_i') = V(p_i) \setminus I_i.\)
To account for this, a new origin-pointer entry is created for \( p_i \), referencing the ID of \( p_i' \) (8 bytes) and encoding the skipped vertices \( I_i \cap V(p_i) \) within the remaining 8 bytes. 
This is achieved by allocating 4 bytes to reference the respective node \( t \) and another 4 bytes as a bitmask, indicating which vertices in \( X_t \) belong to \( I_i \cap V(p_i) \). Notably, this approach imposes a maximum bag size of 32 due to the bitmask representation.

Additionally, within the 16-byte origin-pointer entry, specific bit flags are embedded. Among other purposes, these flags distinguish different solution types (e.g., introduce, join) and play a crucial role in the pruning procedure.

For the surface-pointer entries, a larger 24-byte structure is utilized per solution. Each entry in a surface-pointer file comprises three 8-byte fields: the solution ID, the area, and the perimeter.

\subsection{Pruning}
Once a solution list $D[t,S]$ is no longer needed, we can delete the corresponding surface-pointer file. However, this may leave some entries in the origin-pointer file that are no longer required.
To clean up the origin-pointer file, we use a \emph{pruning} procedure which deletes all unneeded entries. 
This procedure is slow because it operates in place on a file that, during our experiments, could easily exceed a terabyte in size. It needs to identify and retain the necessary entries while reorganizing the data without using additional storage.
Because of this, we try to only execute the pruning if absolutely necessary.
We try to minimize the amount of times the pruning procedure is repeated by estimating if we are likely to run out of storage within a node and only execute the pruning in such cases.
However, this can still lead to unnecessary pruning calls. A possible future approach could allow the algorithm to proceed without prior checks and only revert to the pre-node state once an out-of-storage situation is actually encountered, perform pruning at that point, and then resume execution. This would ensure that pruning is only applied when strictly required.

To estimate whether a node $t$ is at risk of exhausting its storage capacity, we rely on an estimation strategy described in \cref{sec:appendix_choosing_a_td}. Specifically, we estimate the number of solutions that will be generated at the next node. 
This estimate is then used to approximate the number of potential origin-pointer entries and surface-pointer entries being created.

To ensure robustness, we deliberately overestimate these numbers. 
This conservative approach reduces the likelihood to encounter situations where the pruning procedure is necessary but not invoked in time, leading to excessive storage consumption.

One of the primary reasons why the pruning procedure is computationally expensive is that it operates \textit{in-place}. 
We chose this approach because pruning typically occurs when the available hard-disk storage is already critically low. 
In such situations, maintaining additional temporary storage structures would not be feasible due to tight memory constraints.

The in-place pruning process consists of two distinct phases:
In the first phase, we identify and mark all origin-pointer entries that are still required, i.e., entries that are (directly or indirectly) referenced by at least one solution in a surface-pointer file. 
This is achieved by recursively traversing all entries in the origin-pointer file, starting from all solutions stored in the surface-pointer files. 
Dedicated bit flags within each origin-pointer entry are used to indicate whether the entry is still required.

Let \( n \) be the number of entries marked as required in the first phase. 
In the second phase we restructure the origin-pointer file such that all required entries are relocated to occupy the first \( n \) positions of the file.  

To achieve this, we again iterate over all surface-pointer entries and recursively traverse the referenced IDs in the origin-pointer file. 
This traversal is performed in post-order to ensure that dependencies are updated consistently. 
Whenever we encounter a required entry with an ID greater than or equal to \( n \), we move it into a free position among the first \( n \) entries.  

Since relocating an entry overwrites another entry's position, we ensure that only entries that are not required (i.e., unmarked) are overwritten. 

To handle cases where an entry is referenced by multiple other entries, we introduce an additional mechanism: when an entry is moved, we mark its original location with a bit flag, indicating that it has been relocated and store the new ID in that position. 
This allows other referencing entries to correctly resolve the updated position of the entry.

Because moving an entry changes its position in the origin-pointer file, its corresponding ID must be updated. 
To maintain consistency, we process solutions in post-order, ensuring that whenever we relocate an entry for a solution \( p \), we subsequently adjust all references to \( p \) in other entries.

At the end of the second phase we truncate the origin-pointer file to a final size of 
\(n \times 16 \text{ bytes}\), retaining only the required entries.

\subsection{Join node algorithm}\label{apx:join_algorithm}
In \cref{alg:join_heap} we provide a pseudocode to the join algorithm described in \cref{sec:join_heap}.
The input consists of two lexicographically sorted Pareto lists, $P_{t_1}^S=(p_1^1,p_2^1,\dots,p_a^1)$ and $P_{t_2}^S=(p_1^2,p_2^2,\dots,p_b^2)$. 
We initialize an empty min-heap $H$ and insert all solution pairs $(p_i^1,p_1^2)$ for $i \in [a]$ into the heap, i.e., we consider all combinations between solutions from $P_{t_1}^S$ with the first solution from $P_{t_2}^S$.
In the min-heap, solution pairs are always ordered lexicographically. 

We then iteratively extract the solution pair $(p_i^1,p_j^2)$ with the smallest combined area.
If its combined perimeter is smaller than that of the last solution added to the final list $P_t^S$, we know, because of the lexicographic ordering, that this pair corresponds to a new Pareto-optimal solution and add it to $P_t^S$, otherwise the pair is skipped. 
After processing a pair, we remove it from the heap. 
If $j\neq b$, we insert the next solution pair $(p_i^1,p_{j+1}^2)$ into the heap.
In this way, we progressively consider all combinations of solutions from both Pareto lists and only include the ones to $P_t^S$ that are Pareto-optimal.

\begin{algorithm}
\DontPrintSemicolon
  \KwIn{Two lexicographically sorted PO lists $\PP_{t_1}^S=(p_1^1,p_2^1,\dots,p_a^1), \PP_{t_2}^S=(p_1^2,p_2^2,\dots,p_b^2)$, where $p_i^j=(A_i^j,P_i^j)$ encodes the area and perimeter of the solution} 
  \KwOut{$\PP_{t}^S=\PP_{t_1}^S \oplus \PP_{t_2}^S$ (sorted lexicographically)}
  \BlankLine
  Initialize empty min-heap $H$\;
  \For{$p_i^1 \in \PP_{t_1}^S$}{
    Add solution pair $(p_i^1, p_1^2)$ to the min-heap $H$ (sorted lexicographically)
}
$P_{\text{last}} \gets \infty$\;
$\PP_{t}^S \gets \emptyset$\;

\While{$H \neq \emptyset$}{
    $(p_i^1, p_j^2) \gets \text{GetMin}(H)$
    \Comment{Solution pair is judged by its area}\;

    \If{$P(p_i^1) + P(p_j^2) < P_{\text{last}}$}{
        Add the combined solution $(A(p_i^1)+A(p_j^2), P(p_i^1) + P(p_j^2))$ to $\PP_{t}^S$\;
        $P_{\text{last}} \gets P(p_i^1) + P(p_j^2)$\;
    }
    \If {$j == b$ }{ 
    Remove solution pair $(p_i^1, p_j^2)$ from $H$ \Comment{This includes a call of heapify} 
    \;
    }
    \Else{
        Remove solution pair $(p_i^1, p_j^2)$ from $H$\;
        Add solution pair $(p_i^1,p_{j+1}^2)$ to $H$
        \Comment{This includes a call of heapify}\; 
    }
}
\Return $\PP_t^S$\;
\caption{Heap-based algorithm for computing $\PP_{t_1}^S \oplus \PP_{t_2}^S$}
\label{alg:join_heap}
\end{algorithm}

\subsection{Optimizing join node computations}

Let \( t \) be a join node with child nodes \( t_1 \) and \( t_2 \). We consider the case where we need to compute all PO combinations of two lexicographically sorted lists of PO solutions, denoted as \( \PP_1 \) and \( \PP_2 \).

\subsubsection{Computing Lower Bounds}

For each \( i \in \{1,2\} \), we partition \( \PP_i \) into \( n_{\text{lower},i} = \min\{500, \lceil |\PP_i|/10 \rceil\} \) approximately equal-sized consecutive and disjoint sections
$L_1^i,\dots,L_{n_{\text{lower},i}}^i$.
To estimate a lower bound for a section \( L_j^i \), let \( p_1 \) and \( p_2 \) denote the first and last solutions of this section, respectively. 
We define the lower bound as the line connecting the points \( (A(p_1), P(p_1) - \alpha) \) and \( (A(p_2), P(p_2) - \alpha) \), where \( \alpha \) is the smallest value ensuring that all solutions in \( L_j^i \) lie above this line.

\subsubsection{\texorpdfstring{Computing the Heuristic \( \mathcal{H} \) of \( \PP_1 \oplus \PP_2 \)}{Computing the Heuristic H of P1 ⊕ P2}}

To compute \( \PP_1 \oplus \PP_2 \) efficiently, we construct a heuristic Pareto set \( \mathcal{H} \). For each \( i \in \{1,2\} \), we first partition \( \PP_i \) into \( n_{h,i} = \min\{350, \lceil |\PP_i|/10 \rceil\} \) disjoint, consecutive, and approximately equal-sized sections \( H_j^i \) with \( j \in [n_{h,i}] \).

For each section \( H_j^i \), we extract the subset of extreme solutions \( E_j^i \subseteq H_j^i \), yielding a significantly reduced set of representative solutions. 
For every $i\in [2]$, the union \( E_i = \cup_{j} E_j^{i} \) yields a reduced set \( E_i \subseteq \PP_i \) where typically the size of $E_i$ is much smaller than the size of $\PP_i$.

To compute a heuristic approximation \( \mathcal{H} \) of \( \PP_1 \oplus \PP_2 \), we determine all PO combinations of \( E_1 \) and \( E_2 \) using a recursive variant of the join-node speed-up technique. In recursive calls, instead of extracting extreme solutions from new sections, we directly subsample approximately 4\% of the solutions, ensuring an approximately uniform distribution over the set. 
The recursion process continues until a base case is reached where the input sets become sufficiently small, resulting in only one section to be computed per side in the lower-bound calculations. 
At this point, the standard min-heap approach is used to compute all PO combinations.

\subsubsection{\texorpdfstring{Computing Upper Bounds Based on the Heuristic \( \mathcal{H} \)}{Computing Upper Bounds Based on the Heuristic H}}

To derive upper bounds from the heuristic \( \mathcal{H} \), we ensure that any point located above these bounds is guaranteed to be dominated in \( \mathcal{H} \), and therefore also in the true Pareto  set \( \PP_1 \oplus \PP_2 \). For this matter, we construct a set of bounding points \( B \) as follows.

Given the ordered heuristic list \( \mathcal{H} = (p_1, p_2, \dots, p_l) \), we define the bounding points as  
\[
B = \{(A(p_1), P(p_1))\} \cup \{(A(p_2), P(p_1)), (A(p_3), P(p_2)), \dots, (A(p_l), P(p_{l-1}))\}
\]

We then partition \( B \) into \( n_{\text{upper}} = \min\{200, (|B| - 1) / (4-1)\} \) consecutive and overlapping sections \( B_1,\dots,B_{n_{\text{upper}}} \), ensuring that each section contains at least four points. 
Let $B_i=(b_1^i,\dots,b_{l_i}^i)$ be an arbitrary set of points.
The overlap is structured such that for two consecutive sections \( B_i \) and \( B_{i+1} \)
we enforce an overlap at a single point, the last in $B_i$ and the first in $B_{i+1}$, i.e., $b_{l_i}^i = b_1^{i+1}$. This overlapping structure ensures that the upper bounds completely cover the heuristic solution space.

For each section \( B_i \), we define an upper bound as the line segment connecting  
\(
(A(b_1^i), P(b_1^i) + \alpha) \text{ to } (A(b_{l_i}^i), P(b_{l_i}^i) + \alpha)
\)
where \( \alpha \) is the smallest value ensuring that all solutions in \( B_i \) are below this line. 

Figure~\ref{fig:upper_bound_example} illustrates an example for such a section \( B_i \) with only a few points for clarity. Black dots represent the solution points in \( \mathcal{H} \) and red crosses indicate the bounding points \( B_i \). 
The red shaded area indicates a region in which any solution is guaranteed to be dominated. 
The upper bounding line, drawn in green, is constructed by connecting the first and last bounding point in \( B_i \), and shifting this line upwards just enough to ensure that all bounding points lie below or on it. This shift corresponds to the minimal \( \alpha \) value defined as above.

\begin{figure}[t]
    \centering
    \includegraphics[width=0.4\textwidth]{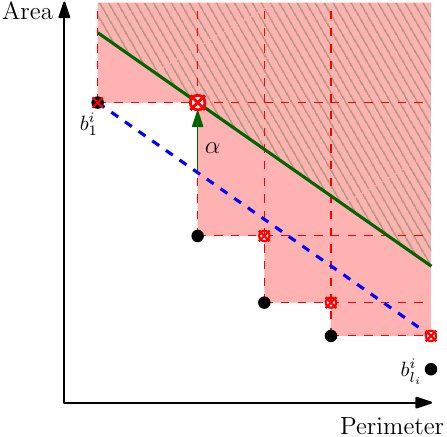}
    \caption{Exemplary upper bound creation. Pareto points $b_1^i,\dots,b_{l_i}^i$ are drawn in black, they dominate all possible points in the red region, and we compute the green line as a heuristic by shifting the blue dotted line up by the minimal amount $\alpha$.}
    \label{fig:upper_bound_example}
\end{figure}

\subsubsection{Identifying Lower Bound Combinations to be ignored}

Let $\mathcal{L}^1 = (L_{1}^1, \dots, L_{m}^1)$ and $\mathcal{L}^2 = (L_{1}^2, \dots, L_{n}^2)$ be two lists of lower bound functions derived from $\PP_1$ and $\PP_2$, respectively. 
For $i\in [2]$, each $L_j^i$ corresponds to a straight line in $\RR^2$, where one axis represents area and the other perimeter. Both lists are sorted in ascending order by area. 
Similarly, let $\mathcal{U} = (U_{1}, \dots, U_{p})$ be a list of upper bounds, also sorted in ascending order by area. 

For each pair $(L_i^1, L_j^2) \in \mathcal{L}^1 \times \mathcal{L}^2$, we construct a polyline $L_{i,j}$, defined by three points derived from $L_i^1$ and $L_j^2$. 
Let $i$ and $j$ be arbitrary but fixed. Let $L_i^1=((A_1^1,P_1^1),(A_2^1,P_2^1))$ and $L_j^2=((A_1^2,P_1^2),(A_2^2,P_2^2))$.
Without loss of generality, we assume the slope ordering
$\frac{P_2^1-P_1^1}{A_2^1-A_1^1} \leq \frac{P_2^2-P_1^2}{A_2^2-A_1^2}$. 
Then, the combined polyline is defined by 
$$((A_1^1+A_1^2,P_1^1+P_1^2),
(A_2^1+A_1^2,P_2^1+P_1^2),
(A_2^1+A_2^2,P_2^1+P_2^2)).$$
This new polyline serves as a lower bound for all possible combinations of solutions lower bounded by $L_i^1$ and $L_j^2$.

Our goal is to determine whether \( L_{i,j} \) lies entirely above all upper bounds \( U_k \in \mathcal{U} \). If this is the case, all solution combinations from \( L_i^1 \) and \( L_j^2 \) can be ignored. 
This is illustrated in \cref{fig:skipping}. Multiple upper bounds are shown as blue lines. In orange, we see a lower bound combination \( L_{i,j} \) that must not be skipped, along with the corresponding solution points. The line combination lies below at least one upper bound and therefore cannot be skipped. In contrast, the green polyline and points represent another lower bound combination that lies entirely above all upper bounds. These solution combinations can safely be skipped.

\begin{figure}[t]
    \centering
    \includegraphics[width=0.4\textwidth]{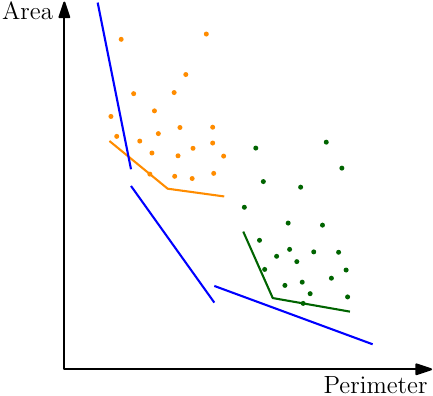}
    \caption{Illustration of the skip condition. The blue line represents the heuristic upper bounds. The orange and green points, along with their respective polylines correspond to combinations of Pareto-optimal solutions and their associated lower bounds.}
    \label{fig:skipping}
\end{figure}

To optimize this check, we leverage the sorted structure of the bounds.
Since all bounds are sorted by area, for each $L_{i,j}$ we only need to consider upper bounds $U_k$ whose area interval overlaps with the interval of $L_{i,j}$. 
As we iterate over $\mathcal{U}$, we identify the first $U_k$ that intersects with the area interval of $L_{i,j}$, check if the actual solution values forming $L_{i,j}$ should be added to the final Pareto set, and then continue with $U_{k+1}$. 
Given that $\mathcal{U}$ is sorted, once we encounter an upper bound $U_l$ that does not overlap with $L_{i,j}$, we can immediately terminate the search, as all subsequent upper bounds will also have non-overlapping intervals.

Since $\mathcal{L}^1$ and $\mathcal{L}^2$ are sorted, we can efficiently determine the starting position in $\mathcal{U}$ for a given $L_{i,j}$. 
Specifically, since the start of the area interval for $L_{i,j}$ cannot be smaller than for $L_{i,j-1}$, we can begin the search for intersecting upper bounds in $\mathcal{U}$ from the first known intersection position of $L_{i,j-1}$. The same principle applies to the outer iteration over $i$, meaning that for $L_{i,1}$, we can reuse the first known intersection position of $L_{i-1,1}$.

\subsection{Choosing a tree decomposition}
\label{sec:appendix_choosing_a_td}
As previously described, we traverse the tree decomposition in postorder and estimate the number of solutions at each node. 
Based on these estimates, we derive predictions for both the runtime and storage consumption of our algorithm. 
To select the most efficient tree decomposition, we minimize a weighted combination of estimated runtime and memory consumption. 
The following sections provide a detailed explanation of the individual components involved in this estimation process.

\subsubsection{Estimating the Number of Solutions per Node}
\label{apx:join_solutions}

To estimate the number of solutions for each node $t$, i.e., $\sum_{S\subseteq X_t} |\mathcal{P}_t^S|$, we begin at the leaf nodes. 
Let $|\PP_t|$ be the number of solutions in node $t$. 
Each leaf node is guaranteed to have exactly one solution. 
For introduce nodes $t$ with child node $t'$, the number of solutions exactly doubles compared to the number of solutions within $t'$. 
Our experiments showed that in forget nodes, the number of solutions approximately halves. 
The behavior of the number of solutions in join nodes is significantly less consistent. 
Let $t$ be a join node with child nodes $t_1$ and $t_2$, and assume $|\PP_{t_1}| \leq |\PP_{t_2}|$. 
Our experiments indicated that the number of solutions at $t$ depend primarily on the larger child $t_2$. 
Consequently, we estimate $|\PP_t|$ using
\[
    |\PP_t| \approx |\PP_{t_2}| \cdot \left( 1 + \left( \frac{|\PP_{t_1}|}{|\PP_{t_2}|} \right)^\alpha \cdot \beta \right),
\]
where $\alpha$ and $\beta$ are learnable parameters. 
Fitting this model using extensive data from our algorithm yields $\alpha \approx 0.57$ and $\beta \approx 2.15$.

\subsubsection{Estimating the Runtime of the Algorithm}

To estimate the runtime of our algorithm, we focus only on the join nodes and their variants (i.e., join-forget and introduce-join-forget nodes), as they dominate the overall runtime. 
Let $t$ be a join node with child nodes $t_1$ and $t_2$, and assume $|\PP_{t_1}| \leq |\PP_{t_2}|$. 
For each $S \subseteq X_t$, the algorithm merges the solution sets $\PP_{t_1}^S$ and $\PP_{t_2}^S$ using a min-heap. 
This yields a runtime of
\(
    \bigO \left( |\PP_{t_1}^S| \cdot |\PP_{t_2}^S| \cdot \log |\PP_{t_1}^S| \right).
\)
This runtime does not account for the effect of the join heuristic.
Our experiments revealed that in the join nodes the heuristic tends to skip a similar proportion of combinations. 
This effect introduces only a linear factor in the time complexity.

We also observed that the sizes of $\PP_t^S$ are approximately uniform across different subsets $S\subseteq X_t$, i.e., $|\PP_t^S| \approx |\PP_t^{S'}|$ for $S \neq S'$. 
Therefore, we approximate
\(
    |\PP_t^S| \approx \frac{|\PP_t|}{2^{|X_t|}}.
\)
Using this estimation, we estimate the runtime for a join-node $t$ as
\(
    |\PP_{t_1}^S| \cdot |\PP_{t_2}^S| \cdot \log |\PP_{t_1}^S| \cdot \gamma,
\)
where $\gamma$ is a learnable constant. Since we are only interested in estimating relative performance among different tree decompositions, the actual value of $\gamma$ is not important.

In join-forget nodes, we observed that the runtime decreases exponentially with the number of forgotten vertices $F$. 
Specifically, we estimate the runtime of join-forget nodes analogously to join nodes, but scale the runtime by approximately $\delta^{|F|-1}$, where $\delta$ is a learnable constant. 
Fitting this runtime approximation to our data, we estimated $\delta \approx 0.5$.

\subsubsection{Estimating Storage Consumption}

To estimate the storage consumption of our algorithm, we focus on the size of the origin-pointer file, as its growth is significantly more problematic than that of the surface-pointer files, whose size has become far less critical due to the incorporation of introduce-join-forget nodes.
The size of the origin-pointer file can be estimated straightforwardly using the previously derived estimates of the number of solutions for each node type.

\subsubsection{Selecting the Optimal Tree Decomposition}

Given estimated runtime and storage consumption for a set of $n$ tree decompositions $\mathcal{T}_1\dots,\mathcal{T}_n$, we select the best-performing tree decomposition by minimizing a performance score $s(\mathcal{T}_i)$. 
This score function is defined as a weighted combination of the normalized runtime and storage estimates:
\[
    s(\mathcal{T}_i) = \frac{1}{4} \frac{\text{estmTime}(\mathcal{T}_i)}{\min_{j\in [n]} \text{estmTime}(\mathcal{T}_j)} + \frac{3}{4} \frac{\text{estmStorage}(\mathcal{T}_i)}{\min_{j\in [n]} \text{estmStorage}(\mathcal{T}_j)},
\]
where $\text{estmTime}(\mathcal{T}_i)$ and $\text{estmStorage}(\mathcal{T}_i)$ denote the estimated runtime and storage requirements for tree decomposition $\mathcal{T}_i$, respectively.

A natural question is how many tree decompositions should be generated before selecting the best one. 
Our experiments showed that our method for generating tree decompositions results in a right-skewed distribution of performance scores: 
Tree decompositions with good scores occur relatively frequently, some exhibit very poor scores. 
To balance quality and preprocessing cost, we generated tree decompositions for approximately 16 hours (i.e., 1 hour using 16 threads) on the same system used in \cref{tab:algorithm_improvements}. 
This approach reduces the risk of selecting suboptimal decompositions while ensuring that the preprocessing time remains relatively low compared to the overall runtime of our algorithm.

\subsubsection{Empirical Evaluation of Tree Decomposition Selection}

In \cref{tab:algorithm_improvements}, we present results demonstrating the impact of selecting an appropriate root and tree decomposition for an early version of our algorithm. 
Further experiments on the final version reinforce the importance of this selection step. 
When using the median-rated tree decomposition from the pool of decompositions used in \cref{tab:algorithm_improvements}, we observed an 87\% increase in runtime and a 275\% increase in storage consumption. 
For the worst-rated tree decomposition, runtime increased by 316\% and storage consumption increased by 16450\%.

\jk{
These results highlight that performance can vary substantially, even if the width stays the same. 
As observed by previous experimental studies~\cite{ADMW15,KKS20}, evaluating a tree decomposition based only on its width does not necessarily yield a good performance.
Especially the number and structure of join nodes are important when evaluating different decompositions for predicting runtime and memory consumption.
}

Motivated by Abseher et al. \cite{ADMW15}, we also evaluate our estimator by the rank attained by the selected tree decomposition within a precomputed pool. They generated $40$ decompositions per instance and sorted them by runtime; a random pick would on average land at rank~$20.5$ and their heuristic achieved a median of $8.5$ on real-world instances of width~$4$–$6$. In subsequent work by Kangas et al.~\cite{KKS20}, it was reported that prediction accuracy decreases as the width increases. 

For all other experiments in this paper we report the full wall-clock time, i.e. including reconstruction of the final solutions. However, in the ranking experiment presented here, we exclude the reconstruction phase for three reasons.
(i)~While the dynamic-programming phase uses up to $96$ threads, reconstruction is strictly single-threaded, so its share in the wall-clock time is disproportionately large even though its CPU usage is minor.
(ii)~As discussed in \cref{sec:appendix_extensions}, our current reconstruction routine is naive; we believe its runtime could be reduced substantially, but this optimization was not required for our study.
(iii)~For each dataset the reconstruction time is almost constant across decompositions, so omitting it barely affects the relative ranking that we are interested in.

For our larger benchmarks we only build one pool per data set (sizes $100$, $40$, and $15$) and record 
(i)~the ranks of the decomposition chosen by our estimator with respect to runtime and storage,
(ii)~the Pearson correlations between our performance score and the actually measured runtime and storage usage, 
(iii)~the speed-up (or memory saving) of the selected decomposition over the median candidate in the pool, and
(iv)~the residual gap to the optimal runtime (storage usage) expressed as the percentage by which the best-observed runtime (or memory usage) still improves on that of the decomposition selected by our estimator.
Results are summarized in \cref{tab:td_ranking} and visualized in \cref{fig:score_scatter}. 
For Lottbek, the six lowest‑scoring decompositions exceeded our benchmark limits, so correlations and median metrics are not reported. The benchmarks were run on the identical hardware configuration described in \cref{sec:appendix_experiments}.

Even for treewidth close to~$20$, the estimator selects decompositions that are well within the top of the ranking.
Compared with the heuristic of Abseher et al., which achieves a median rank of $8.5$ on instances of width~$4$–$6$, our estimator delivers comparable or markedly superior placements while operating at much larger treewidths. The high correlations between score and measured performance also demonstrate that our selection strategy not only scales to wider decompositions but also matches - or exceeds - the effectiveness of the state‑of‑the‑art techniques developed for much narrower treewidth regimes.
Moreover, the right-most columns of \cref{tab:td_ranking} show substantial median improvements: On Osterloh and Ahrem the chosen decomposition is $66\%$ and $56\%$ faster than the median candidate while using $55\%$ and $81\%$ less storage, respectively.
When the estimator did not choose the very best decomposition, the loss was modest: the true optimum is at most $14\%$ faster (and only $2\%$ leaner in storage) in our hardest instance.

\begin{figure}[h]
  \centering
  \subfloat[{Osterloh}: score vs.\ runtime]{
    \includegraphics[width=0.47\linewidth]%
      {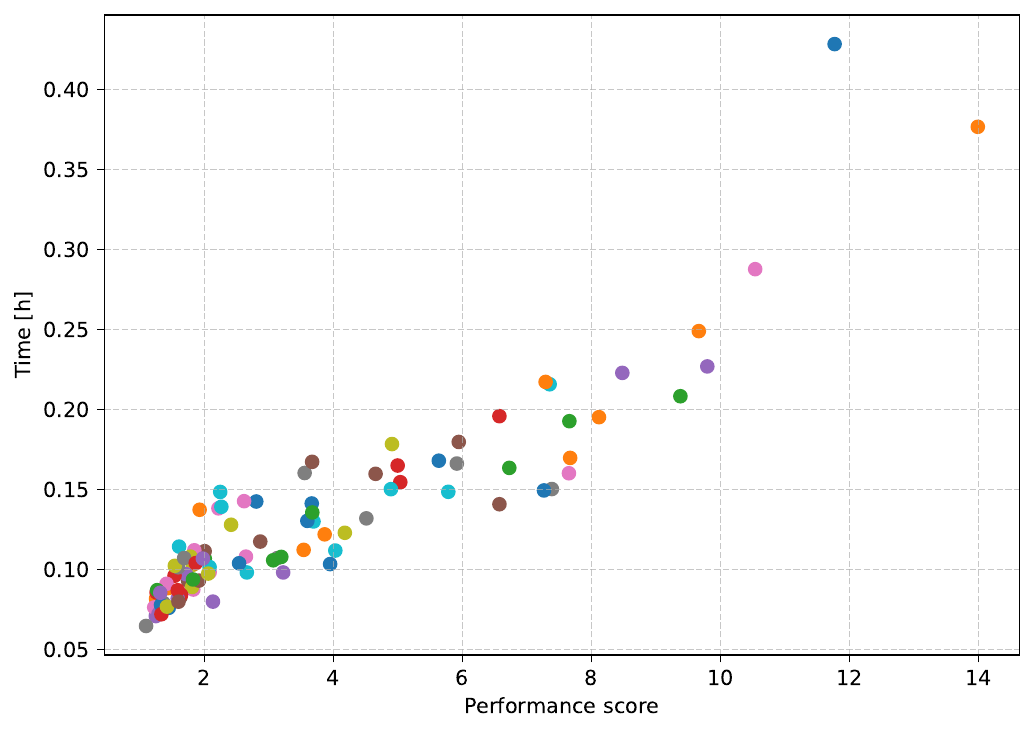}
    \label{fig:osterloh_runtime}}
  \hfill
  \subfloat[{Osterloh}: score vs.\ storage]{
    \includegraphics[width=0.47\linewidth]%
      {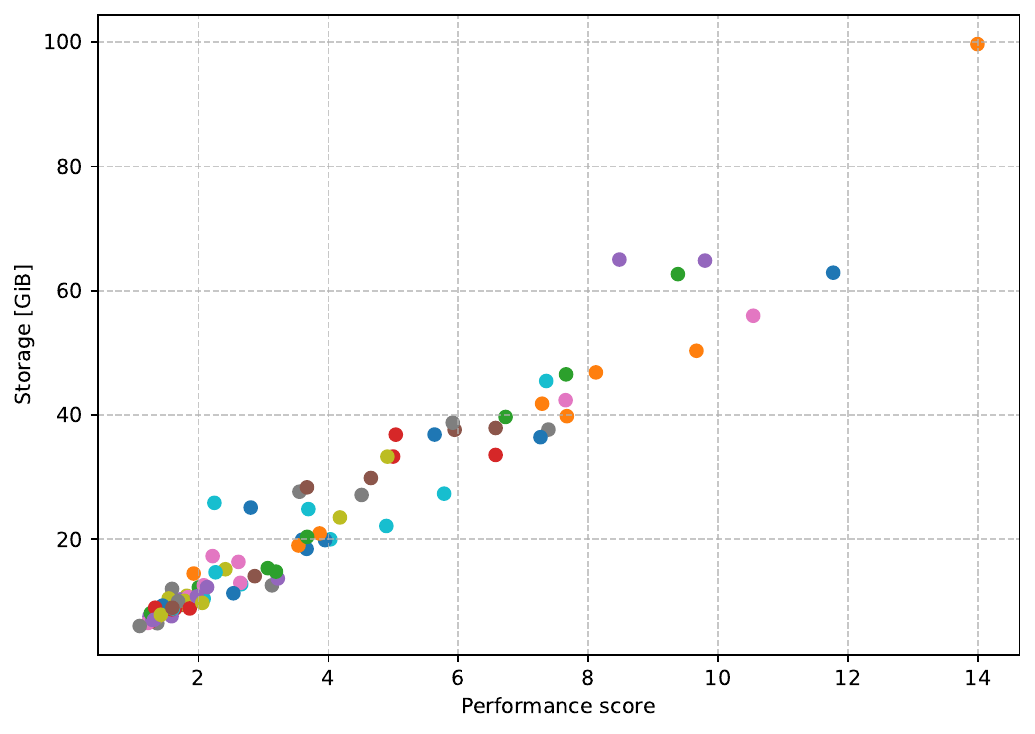}
    \label{fig:osterloh_storage}}
  
  \vspace{0.5em}%
  \subfloat[{Ahrem}: score vs.\ runtime]{
    \includegraphics[width=0.47\linewidth]%
      {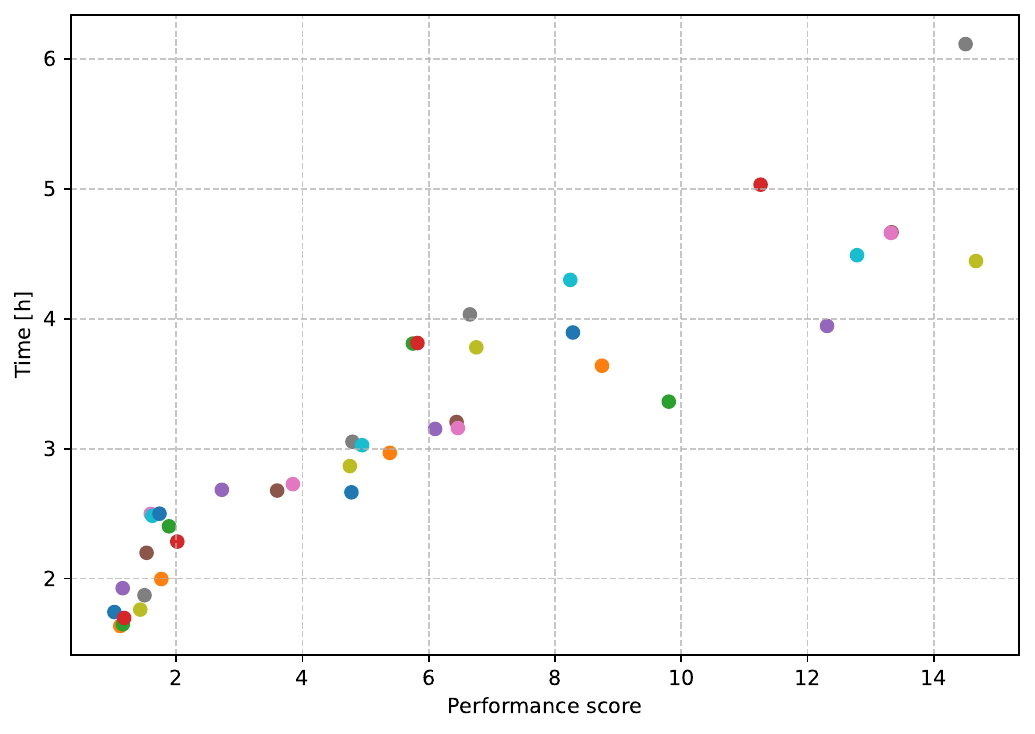}
    \label{fig:ahrem_runtime}}
  \hfill
  \subfloat[{Ahrem}: score vs.\ storage]{
    \includegraphics[width=0.47\linewidth]%
      {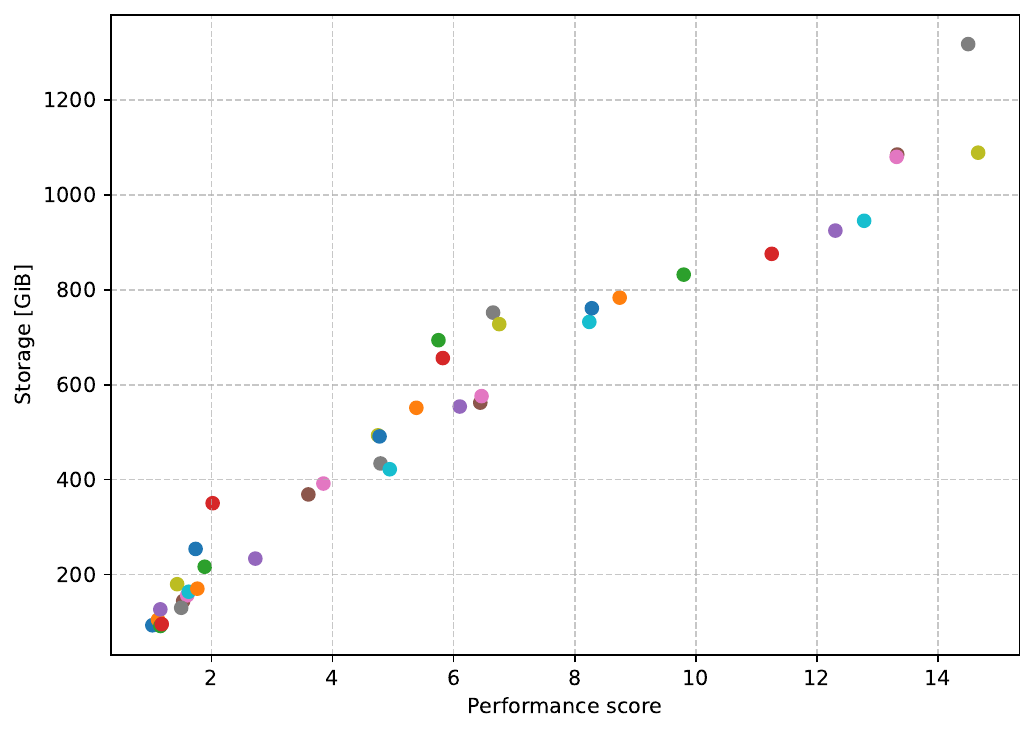}
    \label{fig:ahrem_storage}}
  
  \caption{Relationship between predicted performance score and measured runtime (left) and storage usage (right) for the Osterloh and Ahrem pools.
  Each point represents one tree decomposition; lower values are better
  on both axes. Each tree decomposition is assigned a consistent color across the runtime and storage plots.}
  \label{fig:score_scatter}
\end{figure}

\begin{table}[t]
  \centering
  \begin{tabular}{lccccccccc}
    \toprule
    \multirow{2}{*}{Instance} &
    \multirow{2}{*}{$\Tw$} &
    \multicolumn{2}{c}{Rank} &
    \multicolumn{2}{c}{$\rho$} &
    \multicolumn{2}{c}{Median improvement} &
    \multicolumn{2}{c}{Gap to optimum} \\
    \cmidrule(lr){3-4}\cmidrule(lr){5-6}\cmidrule(lr){7-8}\cmidrule(lr){9-10}
     & & 
     Time & Storage & 
     Time & Storage & 
     Time & Storage & 
     Time & Storage \\
    \midrule
    Osterloh & $13\text{--}14$ & 1/100 & 1/100 & 0.91 & 0.97 & 66\% & 55\% & 0\% & 0\% \\
    Ahrem    & $16\text{--}17$ & 4/40 & 2/40 & 0.92 & 0.98 & 56\% & 81\% & 6\% & 2\% \\
    Lottbek  & 19              & 2/15 & 2/15 & n/a\textsuperscript{*} & n/a\textsuperscript{*} & n/a\textsuperscript{*} & n/a\textsuperscript{*} & 14\% & 2\% \\
    \bottomrule
  \end{tabular}

  \caption{Estimator-selected tree decompositions: pool rank, score–metric correlations, gain over the pool median, and residual gap to the optimum for each benchmark.\\
  *: Six lowest-scoring decompositions exceeded time or memory limits, preventing meaningful correlation estimates or median values.}
    \label{tab:td_ranking}
\end{table}

\section{Further experiments}
\label{sec:appendix_experiments}

In \cref{tab:all_datasets}, we present all datasets that were successfully solved. 
All computations were performed on a high performance computing system.
We used 96 threads of two Intel Xeon "Sapphire Rapids" 2.10GHz and about two terabytes of storage. 
Pruning was disabled for these computations, as the underlying algorithm would have experienced a significant increase in runtime when combined with pruning.

For some datasets we used up to 3200GB of RAM.
This high RAM usage stems from running up to 96 threads in parallel, which, combined with the use of join-forget nodes, caused the algorithm to load a substantial number of solution lists into RAM. 
To mitigate this effect, the algorithm could be adapted to dynamically reduce the number of threads when available RAM is insufficient.

For datasets requiring more than two terabytes of storage, pruning was enabled. 
As a result, execution times increased substantially, with all further experiments taking more than 100 hours. 
We chose a time limit of 100 hours per dataset, so we only evaluated instances that could be solved within this time frame.

For further possibilities to improve runtime and memory consumption, see \cref{sec:appendix_extensions}.

\begin{table}[t]
    \centering
    \begin{tabular}{lrrrrrrrrrr}
        \toprule
        \shortstack[t]{Dataset\\~} & \shortstack[t]{$\Tw$\\~} & \shortstack[t]{Time \\ {[h]}} & \shortstack[t]{Storage \\ {[GIB]}} & \shortstack[t]{\#PO \\ Solutions} & \shortstack[t]{Percentage \\ nonextreme} & \shortstack[t]{Graph \\ \#Vertices} & \shortstack[t]{Graph \\ \#Edges} & \shortstack[t]{ $|V(\mathbb{T})|$ \\~ } \\
        \midrule
        Osterloh        & 14 & 0.15  & 5.8   & 83055 & 99.43 & 1717  & 1887  & 7586   \\
        Schobuell      & 16 & 2.2   & 120   & 193868 & 99.56 & 3552  & 3970  & 16282  \\
        Hambach        & 16 & 4.0   & 73    & 278419 & 99.01 & 7176  & 7571  & 29200  \\
        Ahrem          & 17 & 3.1   & 122   & 219969 & 98.99 & 5280  & 5607  & 21778  \\
        Bockelskamp    & 18 & 3.3   & 171   & 166192 & 99.44 & 2980  & 3325  & 13593  \\
        Gevenich*       & 19 & 4.6   & 214   & 176703 & 99.19 & 4163  & 4552  & 18850  \\
        Sabbenhausen*   & 19 & 5.8   & 253   & 167909 & 99.12 & 4301  & 4651  & 19002  \\
        Lottbek        & 19 & 10.5  & 355   & 316567 & 99.48 & 5595  & 6101  & 24426  \\
        Duengenheim    & 20 & 15.8  & 543   & 281816 & 99.12 & 7024  & 749   & 30468  \\
        Butzweiler     & 21 & 23.9  & 663   & 238057 & 99.44 & 4520  & 5015  & 20296  \\
        Norheim**        & 21 & 39.2  & 1296  & 264172 & 99.53 & 4559  & 5072  & 21370  \\
        Erlenbach      & 22 & 70.7  & 1754  & 303565 & 99.47 & 5644  & 6284  & 25682  \\
        \bottomrule
    \end{tabular}
    \caption{Overview over the datasets we managed to solve. All unmarked datasets required at most 800GB of RAM. Datasets marked with * and ** required at most 1600GB and 3200GB of RAM, respectively.}
    \label{tab:all_datasets}
\end{table}

\subsection{Extrapolating Runtimes}
\label{sec:appendix_extrapolation}
In early versions of our algorithm, the runtime on the Ahrem dataset was prohibitively high, preventing us from solving it entirely. 
To still obtain runtime values for the comparison Table \ref{tab:algorithm_improvements}, we use extrapolated values. 
Specifically, we first ran the final version of our algorithm on this dataset and sampled a fraction of the solution lists that have to be combined at each join node. 
After sampling, we executed the earlier algorithm versions only on these sampled subproblems. 
Since the solution lists within a node are approximately uniform in size, we assumed that the runtime per list remains consistent. 
We extrapolate the total runtime by scaling the measured values proportionally to the number of lists that would have been processed in the full execution.

\section{Implementation Extensions}\label{sec:appendix_extensions}
In this section we discuss potential extensions to our implementation that could further optimize its efficiency and resource consumption. 
We have already outlined possible modifications to the pruning procedure which would invoke it only when necessary (see \cref{sec:implementation_pruning}), and described how RAM usage can be reduced by dynamically adjusting the number of active threads (see \cref{sec:appendix_experiments}). In the following, we propose additional improvements.

\subsection{Caching During Solution Reconstruction}
At the end of the algorithm, the set of solutions (i.e. the selected vertex sets) is reconstructed using the origin-pointer file. 
This process is primarily I/O-bound, making it a significant contributor to the overall runtime for small datasets. 
The impact of this bottleneck is particularly pronounced when multithreading is employed, as the I/O constraints prevent parallel execution during reconstruction, whereas computations at join nodes fully benefit from multithreading.  

The reconstruction procedure currently follows a straightforward, unoptimized approach. 
However, many solutions share references to the same origin-pointer entries. 
By introducing a caching mechanism we could allow frequently accessed entries to be stored in memory, avoiding repeated lookups in the origin-pointer file and thereby reducing redundant I/O operations.
This would likely significantly accelerate the reconstruction phase.  

\subsection{Reducing RAM Usage During Reconstruction}  
Another limitation of the current reconstruction approach is its high RAM consumption. 
All solutions are currently fully reconstructed before being written to an output file. 
This results in substantial memory usage, particularly for large instances with many solutions. 
A more memory-efficient approach would involve writing each reconstructed solution to the output file immediately upon reconstruction, rather than storing all solutions in memory until the process is complete. 
This adjustment would substantially lower peak RAM requirements during the reconstruction process, without negatively impacting performance. 

\end{document}